\documentclass[useAMS,referee,usenatbib]{biom}
\usepackage[figuresright]{rotating}
\usepackage{bm,float}
\usepackage{amsfonts,amsmath,amssymb,bbm}
\usepackage{graphicx}
\usepackage{algorithm}
\usepackage[noend]{algpseudocode}
\usepackage[skip=30pt]{caption}
\usepackage{amsmath,amssymb,mathrsfs,graphicx,url, mathabx}
\usepackage{multirow}
\usepackage[usenames]{color}
\usepackage{ifthen}
\usepackage{setspace}
\usepackage{caption}
\usepackage{subfigure}
\usepackage{float}
\usepackage{graphicx}
\usepackage{algorithm}
\usepackage{bm}
\usepackage[noend]{algpseudocode}
\usepackage{enumitem}
\usepackage{amsfonts}

\theoremstyle{plain}
\newtheorem{thm}{Theorem}

\theoremstyle{definition}

\theoremstyle{remark}

\DeclareMathOperator*{\argmax}{arg\,max}

\newcommand\T{\rule{0pt}{2.6ex}}       % Top strut
\newcommand\B{\rule[-1.2ex]{0pt}{0pt}} % Bottom strut

\def\bSig\mathbf{\Sigma}

\title[ Regularized matrix data clustering and its application to image analysis]{Regularized matrix data clustering and its application to image analysis}

\author{Xu Gao$^{1,*}$\email{xgao2@uci.edu}, 
Weining Shen$^{1,**}$\email{weinings@uci.edu}, and 
Hernando Ombao$^{2***}$\email{hernando.ombao@kaust.edu.sa } \\
$^{1}$Department of Statistics, University of California, Irvine, California, U.S.A. \\
$^{2}$Program on Applied Mathematics \& Computational Science, \\
King Abdullah University of Science and Technology, Saudi Arabia}

\begin{document}
%  This will produce the submission and review information that appears
%  right after the reference section.  Of course, it will be unknown when
%  you submit your paper, so you can either leave this out or put in 
%  sample dates (these will have no effect on the fate of your paper in the
%  review process!)

%\date{{\it Received October} 2007. {\it Revised February} 2008.  {\it
%		Accepted March} 2008.}

%  These options will count the number of pages and provide volume
%  and date information in the upper left hand corner of the top of the 
%  first page as in published papers.  The \pagerange command will only
%  work if you place the command \label{firstpage} near the beginning
%  of the document and \label{lastpage} at the end of the document, as we
%  have done in this template.

%  Again, putting a volume number and date is for your own amusement and
%  has no bearing on what actually happens to your paper!  

%\volume{64}
%\pubyear{2008}
%\artmonth{December}

%  The \doi command is where the DOI for your paper would be placed should it
%  be published.  Again, if you make one up and stick it here, it means 
%  nothing!

%\doi{10.1111/j.1541-0420.2005.00454.x}

%  This label and the label ``lastpage'' are used by the \pagerange
%  command above to give the page range for the article.  You may have 
%  to process the document twice to get this to match up with what you 
%  expect.  When using the referee option, this will not count the pages
%  with tables and figures.  

\label{firstpage}

%  put the summary for your paper here

\begin{abstract}
In this paper, we propose a regularized mixture probabilistic model to cluster matrix data and apply it to brain signals. The approach is able to capture the sparsity (low rank, small/zero values) of the original signals by introducing regularization terms into the likelihood function. Through a modified EM algorithm, our method achieves the optimal solution with low computational cost. Theoretical results are also provided to establish the consistency of the proposed estimators. Simulations show the advantages of the proposed method over other existing methods. We also apply the approach to two real datasets from different experiments. Promising results imply that the proposed method successfully characterizes signals with different patterns while yielding insightful scientific interpretation.  
\end{abstract}

%  Please place your key words in alphabetical order, separated
%  by semicolons, with the first letter of the first word capitalized,
%  and a period at the end of the list.
%


%  As usual, the \maketitle command creates the title and author/affiliations
%  display 

\maketitle
\section{Introduction}
The past decade has witnessed the dramatic progress on technologies that generate high volume datasets. Among them, matrix data is popular and commonly encountered in brain signals and images, e.g., electroencephalography (EEG), functional magnetic resonance imaging (fMRI) and local field potentials (LFPs). In this paper, the goal is to provide a novel framework of analyzing matrix-valued data and apply it to LFPs. %The LFPs essentially capture the integration of membrane currents across local regions of cortex \citep{mitzdorf1985current}. 

Clustering is a fundamental problem in both statistics and machine learning. In this study, we focus on probability model-based clustering approach as it provides likelihood that can be utilized to conduct statistical testing and model selection. In particular, we consider clustering for matrix-valued data. In a motivating example, researchers conducted an olfactory (non-spatial) sequence memory experiment to uncover the neuron learning process on the sequential ordering of odors \citep{allen2016nonspatial}. 12 electrodes were implanted into a rat's brain and LFPs were recorded. The entire experiment consists of 5 odors ABCDE with each corresponding to one epoch. As shown in Figure~\ref{intro_rat}, rats were trained to identify odors denoted by ABCDE with 12 electrodes implanted according to the schematic plot on the right. Preliminary analysis have been conducted to understand the association between the LFPs signals and the particular odor. Figure~\ref{fig:intro} presents the smoothed LFPs across 12 electrodes by different sequence odors and the mean signal. It can be found indisputably that the mean patterns vary dramatically across different odor, which motivates the study of analyzing ``latent" structures. To take one step further, if we compare the signals among different electrodes within each odor, strong spatial dependence can be easily detected. It shows that roughly two ``paradigm" can be found across electrodes especially in Sequence A, B and D. Typical cluster analysis can be done by directly lumping the signals over electrodes as vectors. However, the spatial dependence pattern would be accidently ignored in this case. This innegligible drawback inspires us to develop a statistical strategy directly on the ``matrices" that respect the ``row-wise" and ``column-wise" dependence simultaneuously. 

\begin{figure}[H]
	\centering
	\begin{tabular}{cc}
		\includegraphics[width = 0.45\textwidth, height = 0.2\textheight]{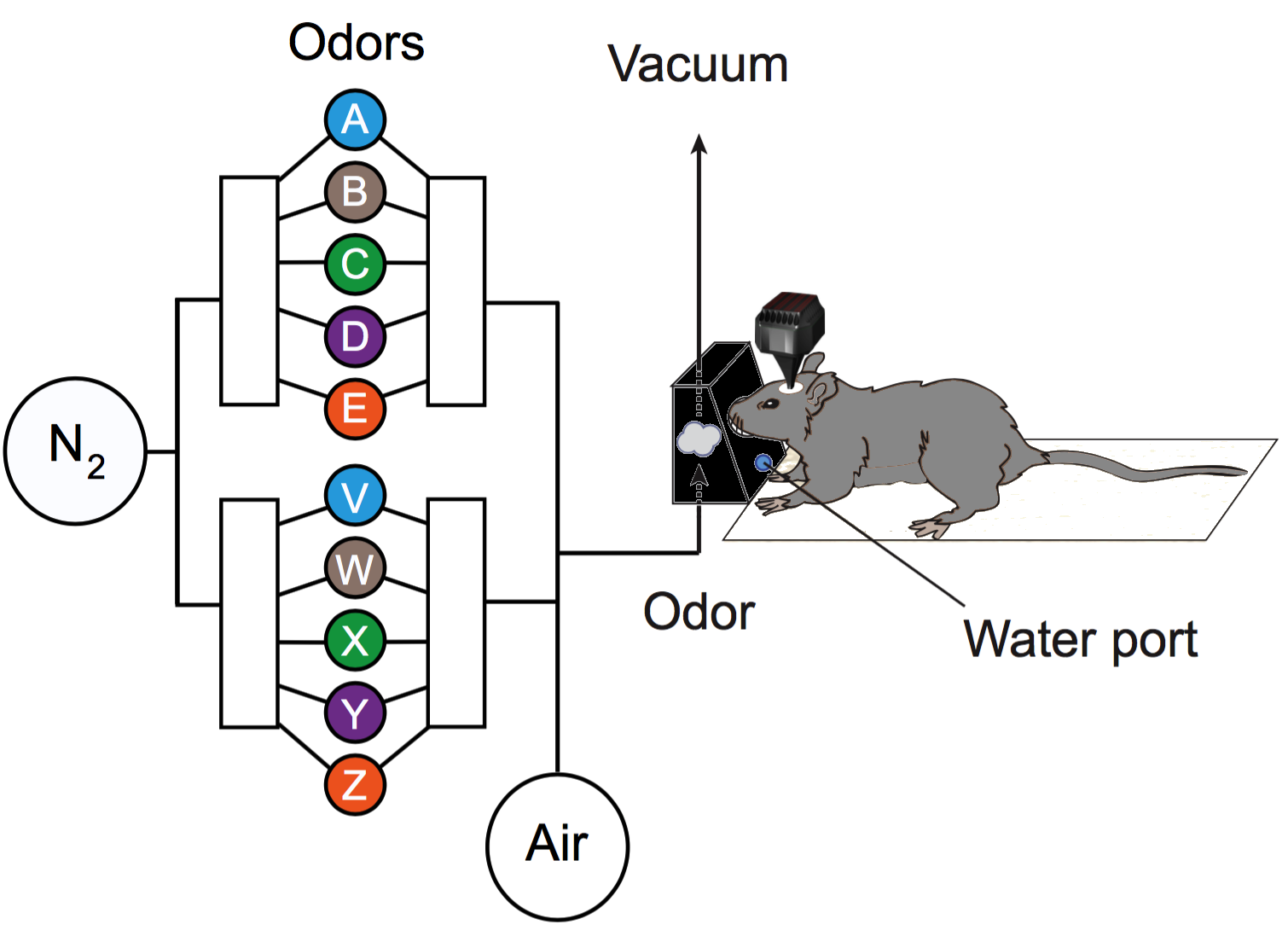} &\	\includegraphics[width = 0.4\textwidth, height = 0.25\textheight]{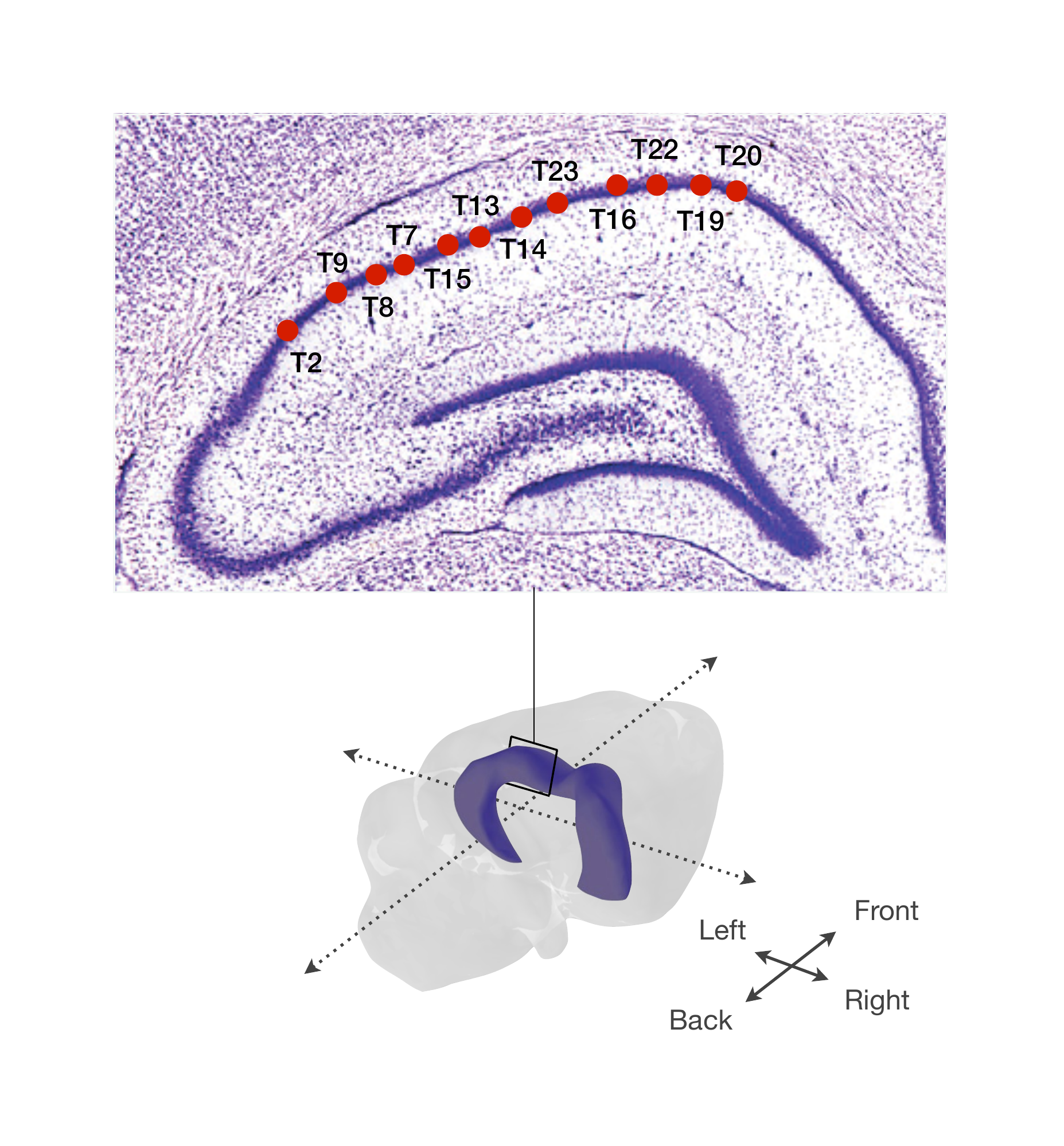}\\ 
		%	\multicolumn{2}{c}{\includegraphics[width = 0.8\textwidth, height = 0.2\textheight]{figures/intro_tetrode7_tsplot}}
	\end{tabular}
	\caption{Left: Apparatus and behavioral design for the olfaction (non-spatial) memory sequence experiment \citep{allen2016nonspatial}. Series of five odors were presented to rats from the same odor port. Each odor presentation was initiated by a nose poke. Rats were tested to correctly identify whether the odor was presented in the correct or incorrect sequence position (by holding their nose in the port until the signal or withdrawing before the signal, respectively). Right: The spatial locations of electrodes implanted in the hippocampus region. The experiment 
		and the data are reported in \citet{allen2016nonspatial}.}
	\label{intro_rat}
\end{figure}
\begin{figure}[H]
	\centering
	\begin{tabular}{cc}
		\includegraphics[width = .5\textwidth, height = 0.25\textheight]{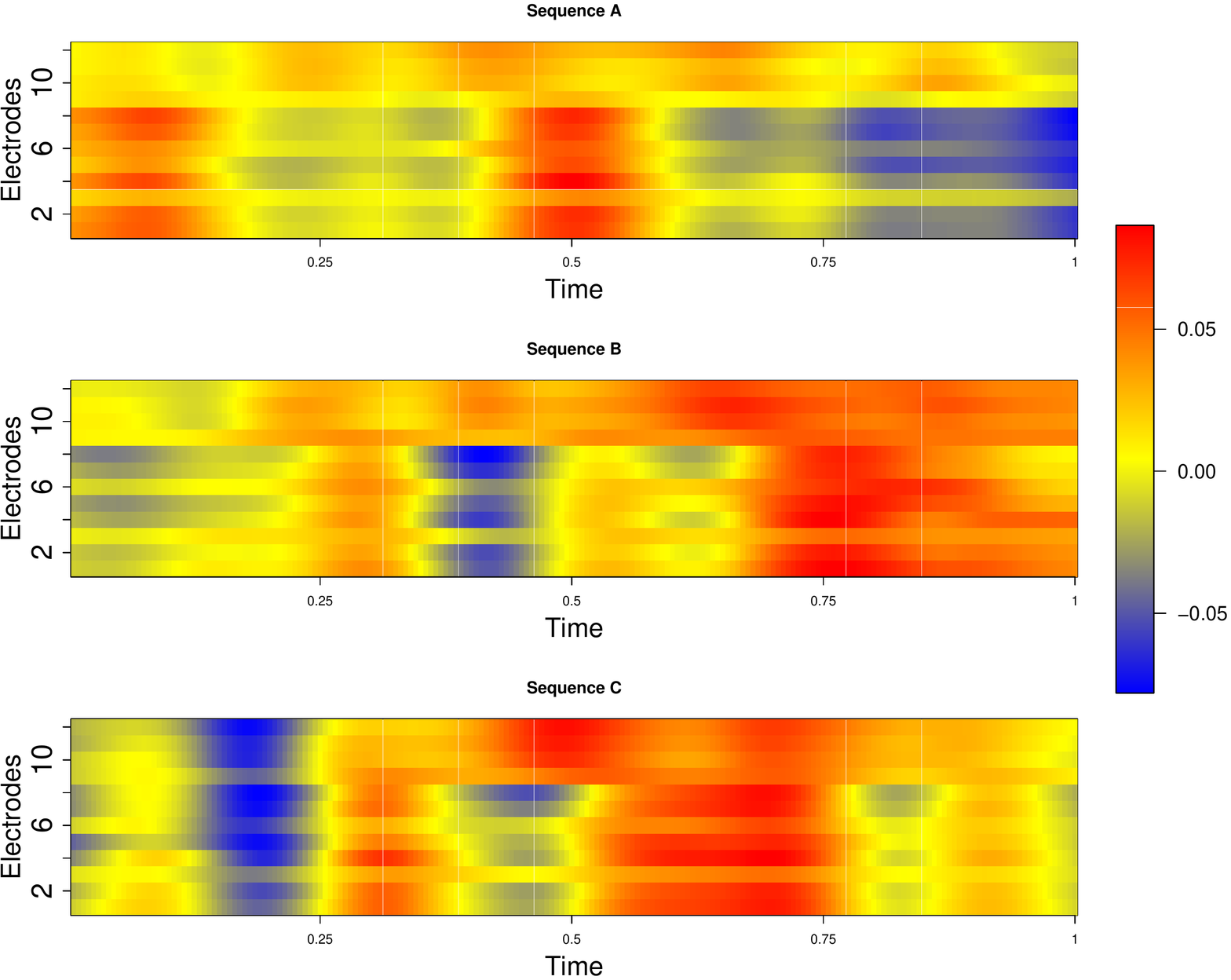}&
		\includegraphics[width = .5\textwidth, height = 0.25\textheight]{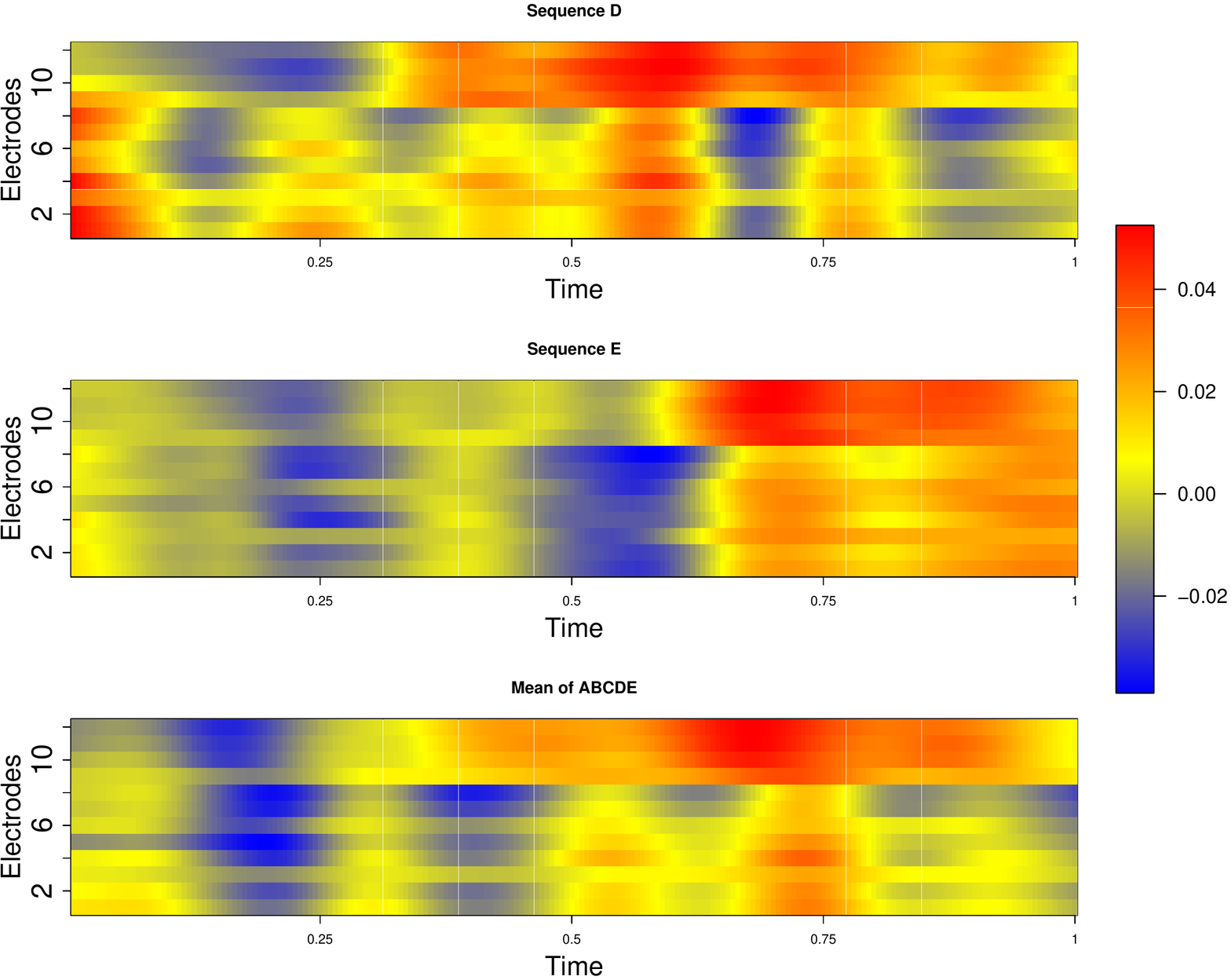}
	\end{tabular}
	\caption{The mean LFPs across different odors.}
	\label{fig:intro}
\end{figure}

Existing mixture models for clustering only handles vectors as variables, where we need to deal with matrix data. The major difficulties are three-fold: (i) we would like to take the matrix (spatial-temporal) structure into account in our modeling process, a simple vectorization of the matrix data will lead to information loss and incorrect interpretability; (ii) there's a common need to impose certain sparsity assumptions on each of the mixture group. Those assumptions are low rank, small valued entries and limited number of non-zero values; (iii) efficient computation and rigorous theoretical justification of the procedure is largely needed for such type of model. 

To solve the aforementioned issues, inspired by the work of \cite{dawid1981some} and \cite{dutilleul1999mle}, we consider a mixture model of matrix normal distributions, whose covariance matrices can be factorized into the Knocker product of two separate column and row covariance matrices. This representation provides both computational convenience and practical interpretation as it separates the variations into time and spatial domains. In addition, we consider three regularization approaches with different norms (e.g. $\ell_1, \ell_2$ and nuclear norm). In terms of computation, we introduce a new EM-type of algorithm that allows multiple regularizations in a unified approach. In theory, we show the strongly consistency of the proposed estimator using the technique in \citep{fan2001variable} with modification to accommodate the matrix-valued data.

%From the literature of statistics and machine learning communities, a large amount of approaches are only applicable to vectors. As shown from the motivating example, such approaches have a few limitations: 1) Spatial and temporal correlation are not easily captured simultaneously; 2) It would be computation demanding when analyzing high-dimensional signals; 3) We would lose the interpretability from the results obtained by the manipulated ``vectors". To address those issues, we propose a probabilistic model directly on the matrix-valued signals. 

%Inspired by the work of \cite{dawid1981some} and \cite{dutilleul1999mle}, the framework is built upon a mixture matrix normal model. For the purpose of clustering signals, the advantages of using such distribution are its interpretability, conceptual and computational easiness. To account for the structures such as sparsity or low-rank, we also introduce flexible regularization terms (e.g. $\ell_1, \ell_2$ and nuclear norm). We have successfully demonstrate that by adding those penalties, the proposed approach outperform over the existing cluster method and also prevent overfitting the training data. On the foundation of the results from \cite{fan2001variable}, we also prove the strongly consistency of the proposed estimator. 

The rest of the paper is organized as follows. In Section~\ref{section:BMND}, we mainly state some background knowledge of matrix normal distribution and the estimation method. In Section~\ref{section:PMMNC}, we introduce the proposed penalized mixture matrix normal model and its estimation approach based on modified Expectation Maximization (EM) and one-step-late algorithms. In Section~\ref{section:Theory}, we provide some theoretic results on the consistency of the (penalized) estimators in a restricted parameter space. In Sections~\ref{section:Simulation}, \ref{section:real1} and \ref{section:real2}, we present some simulation results and apply the proposed method to two LFPs dataset obtained from odor sequence and stroke experiments. 
\section{Background on Matrix Normal Distribution}
\label{section:BMND}
In this section, we mainly focus on a brief review of matrix normal distribution. In the field of modeling image or spatial-temporal data, it is natural to obtain a sequence of matrix valued observations $Y_1, Y_2, \cdots, Y_n$ with dimension $r \times p.$ For example, in the case of multiple spatial-temporal datasets, $p, r$ denotes the spatial and temporal attributes respectively. As an extension of vector-valued data, covariance structures regarding ``spatial" and ``temporal" need to be considered simultaneously. Following the convention of multivariate normal distribution for vectors, $r \times p$ matrix normal distribution $\textit{MN}_{r, p}(M, U, V)$ is defined as 
\begin{equation}
\label{equation:pdf}
f(Y \vert M, U, V) =  \frac{\exp (-\frac{1}{2}\text{tr}(V^{-1}(Y - M)^TU^{-1}(Y - M))}{(2\pi)^{rp/2}|V|^{r/2}|U|^{p/2}},
\end{equation}
where $M \in \mathbb{R}^{r \times p}$, $U \in  \mathbb{R}^{r \times r}$, $V \in  \mathbb{R}^{p \times p}$ and matrices $U$ and $V$ are treated as between and within covariance matrices. With some algebraic manipulations \citep{gupta1999matrix}, it can be shown that 
$Y \sim \textit{MN}_{r, p}(M, U, V)$ if and only if 
\begin{equation}vec(Y) \sim \textit{N}(vec(M), V \otimes U), \label{(dfn)} \end{equation}
where $vec$ is vectorization operation and $\otimes$ is the Kronecker product. It should be pointed that not all the multivariate normal random variable of dimension $r \times p$ is able to convert into matrix normal distribution. Only particular covariance matrices of dimension $rp$ that follow the form in (\ref{(dfn)}) has its corresponding matrix normal representation \citep{dutilleul1999mle}. Such ``separable" \citep{cressie2015statistics} pattern is widely used in the application of electrophysiological data analysis with traditional statistical methods such as state space model \citep{gao2016evolutionary}, vector autoregressive model \citep{derado2010modeling} etc. Moreover, \cite{reinsel1982multivariate} showed it lead to increased estimation accuracy and inferential power when incorporating such structure into analysis. 

\underline{On Estimating the Parameters}

Suppose that $Y_1, Y_2, \cdots, Y_n$ are i.i.d random samples from matrix normal distribution $\textit{MN}_{r, p}(M, U, V)$, the log-likelihood is given by 
\begin{equation}
\ell (M, U, V) = -\frac{npr}{2}\log 2\pi - \frac{nr}{2}\log |V| - \frac{np}{2}\log |U| - \frac{1}{2}\sum_{i=1}^{n}\text{tr}(V^{-1}(Y_i - M)^TU^{-1}(Y_i - M)).
\end{equation}
After some matrix derivatives manipulation, the maximum likelihood estimator (MLE) yields
\begin{align}
\begin{split}
\label{mle_mn}
\widehat{M} &= \sum_{i=1}^{n} Y_i = \bar{Y}\\
\widehat{U} &= \frac{1}{np} \sum_{i=1}^{n} (Y_i - \bar{Y}) \hat{V}^{-1} (Y_i - \bar{Y})'\\
\widehat{V} &= \frac{1}{nr} \sum_{i=1}^{n} (Y_i - \bar{Y})' \hat{U}^{-1} (Y_i - \bar{Y})
\end{split}
\end{align}
It is obvious that there are some identifiability issues since one can simply replace ${U}, {V}$ by $c{U}, \frac{1}{c}{V}$ to satisfy Equations~(\ref{mle_mn}) \citep{dutilleul1999mle}. However, the Kronecker product ${U} \otimes {V}$ will remain invariant and we will mainly focus on the mean parameter $M$ throughout this study. 

There is no close form for $\hat{U}, \hat{V}$. Alternatively, one can utilize iterative algorithms to achieve those values numerically. The algorithm is summarized as follows. Note that this approach is also used as an update step in Section~\ref{section:PMMNC}. 
\begin{algorithm}
	\caption{The MLE of covariance matrices}
	\label{algorithm:covariance}
	\textbf{Input}: $\mathbf{Y} = \{Y_1, Y_2, \cdots, Y_n\}$, $\tau$ (tolerance level), Max-iter\\
	\textbf{Initializing}: 
	$\text{iter} = 0$, $U_0 = \text{diag} (1, \cdots, 1)$, $V_0 = \frac{1}{nr} \sum_{i=1}^{n} (Y_i - \bar{Y})' {U_0}^{-1} (Y_i - \bar{Y})$\\
	$U_1 = \frac{1}{np} \sum_{i=1}^{n} (Y_i - \bar{Y}) {V_0}^{-1} (Y_i - \bar{Y})'$, $V_1 = \frac{1}{nr} \sum_{i=1}^{n} (Y_i - \bar{Y})' {U_1}^{-1} (Y_i - \bar{Y})$\\
	$\mathbf{While}$ (iter $<$ Max-iter or $||U_1 - U_0|| > \tau$ or $||V_1 - V_0|| > \tau$)\\
	$\mathbf{Repeat}\\$
	$U_0:=U_1$\\
	$V_0:=V_1$\\
	$U_1 = \frac{1}{np} \sum_{i=1}^{n} (Y_i - \bar{Y}) {V_0}^{-1} (Y_i - \bar{Y})'$\\
	$V_1 = \frac{1}{nr} \sum_{i=1}^{n} (Y_i - \bar{Y})' {U_1}^{-1} (Y_i - \bar{Y})$\\
	$\text{iter}:= \text{iter} + 1$\\
	$\mathbf{Return}$: $\hat{U}:=U_1$, $\hat{V}:= V_1$
\end{algorithm}
\newpage
\begin{remark}
Note that $||.||$ denotes the Frobenius norm. $\text{diag} (1, \cdots, 1)$ represents the identity matrix of dimension $r$.
\end{remark}

\section{Penalized Mixture Matrix Normal Clustering}
\label{section:PMMNC}
\subsection{Mixture Matrix Normal Models}
\label{section:MMNM}
Suppose the observed matrix-valued data $Y_1, \cdots, Y_n$ are obtained from a population with $k$ ``regimes". The probability density function is essentially a mixture of matrix normal densities. For simplicity, if we write $\Theta_j = (M_j, U_j, V_j),$ and the prior association densities as $\pi_j, j=1, \cdots, k,$ then the marginal density function of $Y_i$ can be written as 
\begin{equation}
f(Y_i \vert \Theta_1, \cdots, \Theta_k, \pi_1, \cdots, \pi_k) = \sum\limits_{j=1}^k\pi_jf(Y_i\vert \Theta_j),
\label{MMNM:mixturepdf}
\end{equation}
where $f(Y_i\vert \Theta_j)$ is shown in Equation~(\ref{equation:pdf}) and $\sum\limits_{j=1}^k\pi_j=1$.
The log-likelihood yields
\begin{equation}
\ell_{obs} (\Theta_1, \cdots, \Theta_k, \pi_1, \cdots, \pi_k)= \sum\limits_{i=1}^n \log \{ \sum\limits_{j=1}^k\pi_jf(Y_i\vert \Theta_j)\}.
\label{MMNM:obslikelihood}
\end{equation}
\underline{On Estimating the Parameters}

Expectation Maximization (EM) algorithm \citep{dempster1977maximum} can be efficiently used to estimate the parameters. In general, it is an iterative approach consisting of expectation (E) and maximization (M) steps. 

In the E-step, a posterior probability of observation $Y_i$ belongs to the $j-th$ cluster is calculated by Bayes Theorem that 
\begin{equation}
\label{equation:estep}
\alpha_{ij} = \frac{\pi_j f(Y_i|\Theta_j)}{\sum\limits_{l=1}^k\pi_l f(Y_i|\Theta_l)}.
\end{equation}
In the M-step, the estimates of the parameter vector are obtained by solving the non-constraint optimization problem that 
\begin{equation*}
\widehat{\Theta}_j = \argmax_{\Theta_j} \sum\limits_{i=1}^n\sum\limits_{j=1}^k\alpha_{ij}\log \{\pi_jf(Y_i\vert \Theta_j)\}
\end{equation*}
After some matrix derivatives and algebra manipulations, we can obtain the explicit solutions that 
\begin{align}
\label{equation:mstep}
\begin{split}
\hat{\pi}_j &= \frac{\sum_{i=1}^{n}\alpha_{ij}}{n}\\
\widehat{M}_j &= \frac{\sum_{i=1}^n\alpha_{ij}Y_i}{\sum\limits_{i=1}^n\alpha_{ij}}\\
\widehat{U}_j &= \frac{\sum_{i=1}^n\alpha_{ij}(Y_i - \widehat{M}_j)\widehat{V}_j^{-1}(Y_i - \widehat{M}_j)'}{p\sum\limits_{i=1}^n\alpha_{ij}}\\
\widehat{V}_j &= \frac{\sum_{i=1}^n\alpha_{ij}(Y_i - \widehat{M}_j)'\widehat{U}_j^{-1}(Y_i - \widehat{M}_j)}{r\sum\limits_{i=1}^n\alpha_{ij}}
\end{split}
\end{align}
Note that $\widehat{U}_j, \widehat{V}_j$ can be obtained numerically using the similar method to Algorithm~\ref{algorithm:covariance}. 

Although mixture matrix models are widely used in high dimensional data clustering analysis, they neglect the sparsity structures that are commonly encountered in applications, e.g. the illustrative example shown in Section~\ref{intro_rat}. We propose our penalized mixture model to account for this limitation in the following section. 

\subsection{Penalized Mixture Matrix Normal Models} 
It is quite common that we have some prior information on parameters $\Theta$. This could originate from the sparsity, rank, smoothness or a prior probability density on parameters \citep{green1990use}, which could simplify interpretation or the parametric structures.  %\textcolor{red}{Need to add some reference and motivation here: the mean $M_j$ are of low rank, sparse or small values.} 
Specifically, the mean signals $M_j$ could be of low rank, sparse or small values. To this end, it is natural to add a regularization term to the likelihood and alternatively, maximum penalize likelihood estimate should be obtained.  Specifically, we define the penalized log-likelihood as
\begin{equation}
Q (\lambda, \Theta_1, \cdots, \Theta_k, \pi_1, \cdots, \pi_k)= \sum\limits_{i=1}^n \log \{ \sum\limits_{j=1}^k\pi_jf(Y_i\vert \Theta_j)\} - \lambda P(\Theta),
\label{eq:obspenalized}
\end{equation}
where $P(.)$ is some penalized function. Examples can be logarithm of probability density functions, $\ell_1, \ell_2$ norms, nuclear norm etc. 

\underline{On Estimating the Parameters}

Similar to the approach in Section~\ref{section:MMNM}, we propose a modified EM algorithm to estimate the parameters. The E-step can be easily achieved by Equation~(\ref{equation:estep}). The M-step boils down to the optimization problem where
\begin{equation}
\widehat{\Theta} = \argmax_{\Theta} \sum\limits_{i=1}^n\sum\limits_{j=1}^k\alpha_{ij}\log \{\pi_jf(Y_i\vert \Theta_j)\} - \lambda P(\Theta).
\end{equation}
In contrast to the case without penalty, the solution $\hat{\Theta}$ may not have an explicit form. \cite{lange1995gradient} proposed a gradient method related to EM algorithm. It replaces the M-step by conducting one iteration of Newton's method. Theoretic results on the convergence were also discussed. As an alternative approach, other methods including surrogate functions \citep{lange2000optimization}, overrelaxed EM algorithm \citep{yu2012monotonically} were introduced to this issue. 

Throughout this article, we mainly focus on three types of penalties: $\ell_1, \ell_2$ and nuclear norm. \cite{pan2007penalized} introduced $\ell_1$ penalty to the mean parameters in the setting of mixture univariate normal models. An explicit form of the M-step is derived using a sub-gradient. \cite{green1990use} developed  the ``one-step-late" (OSL) algorithm that can be applied to more general case. Inspire by the aforementioned results, we developed a sub-gradient update for $\ell_1$ norm and  OSL step for $\ell_2$ and nuclear norms. 

In the case of $\ell_1$ norm penalty, the update of $M_j$ is the optimal value that maximizes 
\begin{equation*}
\sum\limits_{i=1}^n\sum\limits_{j=1}^k\alpha_{ij}\log \{\pi_jf(Y_i\vert \Theta_j)\} - \lambda \sum\limits_{j=1}^k||M_j||_1.
\end{equation*}
Following a similar derivation by \cite{pan2007penalized}, the update step of $M_j$ has the form that

\begin{equation}
\label{equation: l1penalty}
\hat{M_j} = \text{sign}(\tilde{M}_j)(|\tilde{M}_j| - \frac{\lambda}{\sum_{i=1}^{n}{\alpha_{i,j}}} U_i  \mathbf{1}_{r\times p} V_i)_{+},
\end{equation}
where $\tilde{M}_j = \frac{\sum_{i=1}^{n}{\alpha_{i,j} Y_i}}{\sum_{i=1}^{n}{\alpha_{i,j}}}$ is the update for $M_i$ without penalty, $B_+ = max(B, 0)$, $\mathbf{1}_{r\times p}$ is a matrix of all 1's. sign() and $(.)_+$ are all component-wise operators. 

In the case of $\ell_2$ norm penalty, the objective function is derived to be 
\begin{equation*}
Q_{\ell_2} (\pi, \Theta)= \sum\limits_{i=1}^n\sum\limits_{j=1}^k\alpha_{ij}\log \{\pi_jf(Y_i\vert \Theta_j)\} - \lambda \sum\limits_{j=1}^k||M_j||_2.
\end{equation*}
After matrix derivative manipulations, we have
\begin{equation*}
\frac{\partial Q_{\ell_2} (\pi, \Theta)}{\partial M_j} = U_j^{-1}\sum_{i=1}^{n}{\alpha_{i,j}(Y_i - M_j)V_j^{-1}} - 2\lambda M_j,
\end{equation*}
The update step of $M_i$ follows the form
\begin{equation}
\label{equation:l2penalty}
\hat{M}_j = \tilde{M}_j - \frac{2\lambda}{\sum_{i=1}^n\alpha_{ij}} U_j M_j V_j,
\end{equation}
where $U_j, M_j, V_j$ are the update from the previous step. 

For the case of $nuclear$ norm penalty, similar derivation yields 
\begin{equation}
\label{equation:nuclearpenalty}
\hat{M}_j = \tilde{M}_j - \frac{\lambda}{\sum_{i=1}^n\alpha_{ij}} U_j \Phi_j \Omega_j' V_j,
\end{equation}
where $M_j$ has the singular value decomposition $M_j = \Phi_j \Lambda_j\Omega_j' $.

As a summary, the proposed estimation approach involves algorithms of initialization and alternating from E-step and M-step. Details are presented as follows

\noindent {\bf I. (Initialization)} We start with vectorizing the original matrix-valued observations $Y_1, \cdots, Y_n$ and applied k means to achieve the initial cluster membership values, written as $S_1, \cdots, S_k$, where $S_j = \{i ~|~Y_i ~\text{in j-th cluster}\}$. Note that we can relax this step by randomly assign clusters to those observations. Then for each cluster, the initial value of $\Theta_i$ can be obtained following the same manner in Section~\ref{section:BMND}. $\pi_j$ can be directly estimated by $\hat{\pi}_j = \frac{|S_j|}{n}$

\noindent {\bf II. (E-step)} 
We update the posterior membership by 
$$
\alpha_{ij} = \frac{\pi_j f(Y_i|\Theta_j)}{\sum\limits_{l=1}^k\pi_l f(Y_i|\Theta_l)}.
$$
\noindent {\bf III. (M-step)} 
The mean parameter $M_j$ with respect to various penalties is updated by the Equations~\eqref{equation: l1penalty}, \eqref{equation:l2penalty} and \eqref{equation:nuclearpenalty} respectively. Updates for $\pi_j, U_j, V_j$ follows Equations~\eqref{equation:mstep} and Algorithm~\ref{algorithm:covariance} is also utilized. 

\noindent {\bf IV. (Stopping criteria)} The iterative approach will alternate by {\bf I.} and {\bf II.} until certain iterations have been reached or the Frobenius norm change of the mean parameter $M_j$ is small enough. 

\underline{On Choosing the Number of Clusters}

A key question in the proposed method is to determine the number of clusters. Inspired by the approach proposed by \cite{smyth2000model}, we introduce cross validated penalized likelihood (CVPL) as the key measure. Without loss of generality, let us denote $f(.), f_k(.)$ as the ``true" and $k$ mixture probability density functions, $\Psi, \Psi_k$ as the corresponding parameters. We split the dataset $\bm{Y} = \{Y_1, \cdots, Y_n\}$ into training and testing groups denoted by $\bm{Y}_{train}, \bm{Y}_{test}.$ If we write the averaged penalized negative log-likelihood as 
\begin{equation}
\ell_k = -\frac{1}{N_{test}} \left(\ell_{obs}(\Psi_k(\bm{Y}_{train})|\bm{Y}_{test})  - \lambda P(\Psi_k) \right)
\end{equation}
It can be shown directly that 
\begin{equation}
E(\ell_k) = \int \log \frac{f(\bm{Y})}{\tilde{f}_k(\bm{Y})} f(\bm{Y}) d\bm{Y} + C,
\end{equation}
where $\tilde{f}_k(\bm{Y}) = \exp \{ \log f_k(\bm{Y}) -  \lambda P(\Psi_k)\}.$
It shows that the expectation of $\ell_k$ is the Kullbak-Leibler (KL) distance between $f(.)$ and the exponential penalized k mixture likelihood up to some constant. Derived from this result, we propose CVPL to determine the optimal number of clusters.

\section{Theory}
\label{section:Theory}
In this section, we first show some theoretic results on the consistency of the maximum likelihood estimator without regularizations. In order to guarantee a constrained (global) maximum likelihood formulation, we define the constrained parameter space $\Psi^{d_1, d_2}$ as 
\begin{align}
\begin{split}
&\Psi^{d_1, d_2} = \{\pi_1, \cdots, \pi_k \in \mathbb{R},~ M_1, \cdots, M_k \in \mathbb{R}^{r * p}, ~V_1 \otimes U_1, \cdots, V_k \otimes U_k \in \mathbb{R}^{rp * rp}: \\
&\min_{1\leq h \neq j \leq k} \rho (U_hU_j^{-1}) \geq d_1 > 0, 
\min_{1\leq h' \neq j' \leq k} \rho (V_{h'}V_{j'}^{-1}) \geq d_2 >0,~\sum_{i=1}^{k} \pi_i = 1, ~\pi_l > 0,  \\
&\rho(U_l)>0 , ~\rho(V_l) > 0~\text{for}~l = 1, \cdots, k\},
\end{split}
\label{theory:condition}
\end{align}
where $d_1, d_2 \in (0, 1]$, $\rho(.)$ denotes the minimum eigenvalue.  
\begin{thm}
	Let $Y_1, \cdots, Y_n$ be random samples from a mixture matrix normal distribution~\eqref{MMNM:mixturepdf}, then for $d_1, d_2 \in (0, 1]$, there exists a constrained global maximizer $\hat{\psi}^n$ of the log-likelihood~\eqref{MMNM:obslikelihood} over $\Psi^{d_1, d_2}$. Moreover, $\hat{\psi}^n$ is also strongly consistent in $\Psi^{d_1, d_2}$. 
	\label{thm:nopenalty}
\end{thm}
\begin{proof}
	First, we state the fact that 
	\begin{equation}
	\min_{1\leq h \neq j \leq k} \rho(\Sigma_h \Sigma_j^{-1}) \geq \min_{1\leq h \neq j \leq k} \rho(V_h V_j^{-1}) * \min_{1 \leq h' \neq j' \leq k} \rho(U_{h'} U_{j'}^{-1}), 
	\label{theory:fact}
	\end{equation}
	where $\Sigma_h = V_h \otimes U_h$.
	
	Actually, it follows directly from the property that 
	\begin{align*}
	\rho(\Sigma_h \Sigma_j^{-1}) & = \rho \big[(V_h \otimes U_h) (V_j \otimes V_h)^{-1} \big]\\
	& = \rho \big[ (V_h V_j^{-1}) \otimes (U_h U_j^{-1}) \big]\\
	& =  \rho(V_h V_j^{-1}) * \rho(U_h U_j^{-1}), 
	\end{align*}
	%where $\rho(V_h V_j^{-1}) * \rho(U_h U_j^{-1}) = \{  \kappa: \zeta \in \rho(V_h V_j^{-1}), \kappa \in \rho(U_h U_j^{-1}) \}$ and 
	where the equalities follow the results in \cite{schacke2004kronecker}.
	We denote the parameter space $\widetilde{\Psi}^d$ as 
	\begin{align}
	\begin{split}
	&\widetilde{\Psi}^d = \{\pi_1, \cdots, \pi_k,~ M_1, \cdots, M_k, ~V_1 \otimes U_1, \cdots, V_k \otimes U_k: \min_{1\leq h \neq j \leq k} \rho (\Sigma_h\Sigma_j^{-1}) \geq d > 0,  \\
	& ~\sum_{i=1}^{k} \pi_i = 1, ~d_1 d_2 = d, ~\pi_l > 0, ~\rho(\Sigma_l)>0 ~\text{for}~l = 1, \cdots, k\},
	\end{split}
	\end{align}
	then due to the definition \eqref{(dfn)} and results in \citep{hathaway1985constrained}, there exists a global constraint maximizer of \eqref{MMNM:obslikelihood} $\hat{\psi}^n$ over $\widetilde{\Psi}^d$ so that $\ell_{obs} (\hat{\psi}^n)= \sup\limits_{\widetilde{\Psi}^d} \ell_{obs} (\psi)$ and there exists a compact set $S \in \widetilde{\Psi}^d$ such that $\hat{\psi}^n \in S$ and $\sup\limits_S \ell_{obs} (\psi) = \sup\limits_{\widetilde{\Psi}^d} \ell_{obs} (\psi).$ Moreover, the fact \eqref{theory:fact} implies that $\sup\limits_{\widetilde{\Psi}^d} \ell_{obs} (\psi) \geq \sup\limits_{{\Psi}^{d_1, d_2}} \ell_{obs} (\psi)$ for any $d_1, d_2.$ Due to the boundedness of $S$, it can be shown by contradiction that there exist $d_1, d_2$ so that $S \in \Psi^{d_1, d_2}.$ Thus, we have that 
	$\sup\limits_S \ell_{obs} (\psi) = \sup\limits_{\widetilde{\Psi}^d} \ell_{obs} (\psi) \geq \sup\limits_{{\Psi}^{d_1, d_2}} \ell_{obs} (\psi) \geq \sup\limits_S \ell_{obs} (\psi),$ which completes the proof of the first part. To show the strongly consistency, the same argument can be utilized as in \cite{hathaway1985constrained} with the fact of definition \eqref{(dfn)}.
\end{proof}

\begin{remark}
	Note that the preceding results hold for unidentifiable case resulting from \cite{hathaway1985constrained}.	
\end{remark}
\begin{remark}
	The condition in \eqref{theory:condition} is not easy to check in practice. One might bound all the eigenvalues within an interval $(a, b)$ for numerical stability.
\end{remark}
Next, we will show that under wild conditions, there also exists a root-n consistent penalized likelihood estimator of \eqref{eq:obspenalized}. We first define the parameter space denoted as $\widebar{\Psi}^{d_1, d_2}$ where
\begin{align}
\begin{split}
\widebar{\Psi}^{d_1, d_2} = \{& \pi_1, \cdots, \pi_k, M_1, \cdots, M_k, V_1 \otimes U_1, \cdots, V_k \otimes U_k \in \Psi^{d_1, d_2} : \frac{\sigma_i(U_h)}{\sigma_i(V_h)} = c_h~\text{for}~i = 1, \cdots, \min\{r, p\} \\
& h = 1, \cdots, k \},  
\end{split}
\label{condition:identifiability}
\end{align}
where $\sigma_i(U_h)$ denotes the $i$th eigenvalue of matrix $U_h$ and $c_h$ is a positive constant. 
\begin{thm}
	Let $Y_1, \cdots, Y_n$ be random samples from a mixture matrix normal distribution~\eqref{MMNM:mixturepdf},  in the case of $\ell_1$ and $\ell_2$ norm penalties, under condition (A) (in the appendix), if $\lambda  = O_p(n^{\eta}), ~0 < \eta \leq \frac{1}{2}$, then there exists a local maximizer $\hat{\zeta}$ of the penalized likelihood $\eqref{eq:obspenalized}$ such that $||\hat{\zeta} - \psi_0|| = O_p(n^{-1/2})$ in the parameter space $\widebar{\Psi}^{d_1, d_2},$ where $\psi_0$ is the true parameter in $\widebar{\Psi}^{d_1, d_2}.$
	\label{thm:penalty}
\end{thm}
\begin{proof}
	The proof can be directly adapted from the argument of Theorem 1 proposed by \cite{fan2001variable}. It suffices to check the conditions in their proof. For the first condition, all the assumptions are trial except the identiability issue. Actually, since $\sigma_i(V_h\otimes U_h) = \sigma_{i'}(V_h) \sigma_{i''}(U_h)$, by fixing the ratio of eigenvalues as shown in \eqref{condition:identifiability}, there exists a unique eigenvalue pair of $\sigma_{i'}(V_h), \sigma_{i''}(U_h)$ for a given value of $\sigma_i(V_h\otimes U_h).$ Thus $V_h \otimes U_h = V_h' \otimes U_h'$ implies $V_h = V_h'$ and $U_h = U_h'.$ The identifiablity property then directly follow given the results from \cite{yakowitz1968identifiability}. For the second condition, our assumption (A) directly implies that. For the last condition, it holds from the compactness of the parameter space $\widebar{\Psi}^{d_1, d_2}.$
\end{proof}

\section{Simulations}
\label{section:Simulation}
\subsection{Results on Choosing the Number of Clusters}
\label{section:sim1}
In this section, we evaluate the effectiveness of the proposed cross validated penalized likelihood (CVPL) in different scenarios. We generate two clusters of signal that follow matrix normal distribution with mean structures shown in Figure~\ref{fig:sim1_squarecross}. The row-wise and column-wise covariance matrices follow an autoregressive setting where ${\it cov} \{Y_{k_1, l_1}, Y_{k_2, l_2} \} = 0.9^{|k_1 - k_2| + |l_1 - l_2|}, 1 \leq k_i \leq r, 1 \leq l_i \leq p.$ The proportion for both of the clusters is equal. In Scenario I, we set the number of signals $n = 100$ with $r = p = 60.$ In Scenario II, we let $n = 50, ~r=p=30$. $200$ simulations were conducted for each of the two cases. 
\begin{figure}[H]
	\centering
	\begin{tabular}{cc}
		\includegraphics[width = .45\textwidth, height = 0.25\textheight]{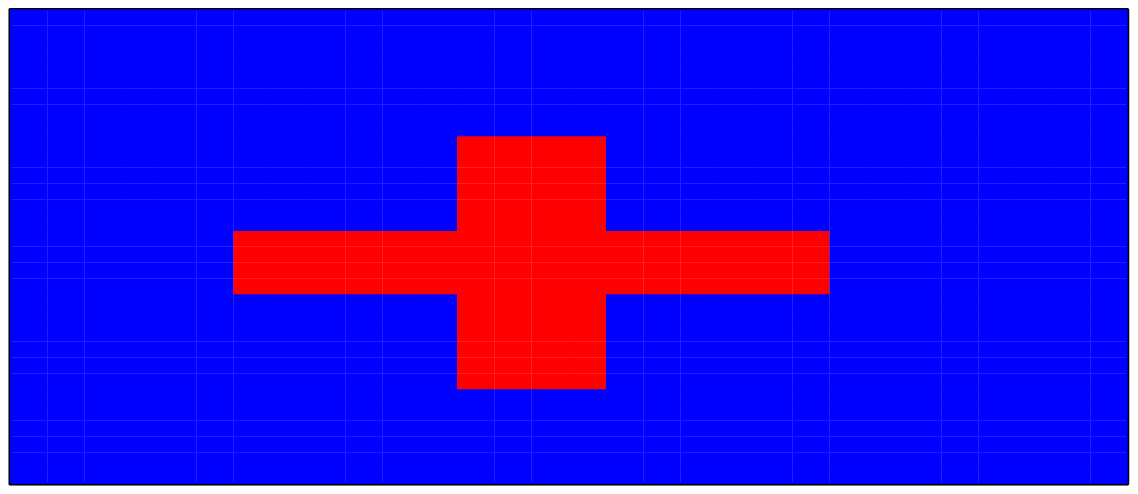}&
		\includegraphics[width = .45\textwidth, height = 0.25\textheight]{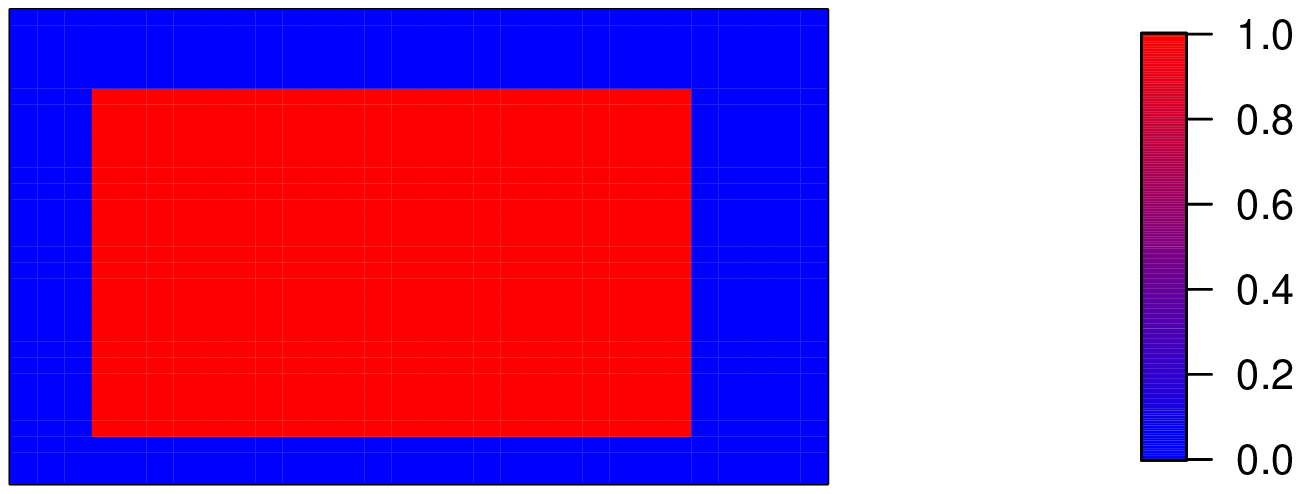}
	\end{tabular}
	\caption{The mean structure of the two clusters.}
	\label{fig:sim1_squarecross}
\end{figure}
We applied the proposed method to the simulated dataset. $L_1, L_2$ and Nuclear penalties were all implemented. As is shown in Table~\ref{sim1:table1}, among all the penalties, $\lambda$ and sample sizes, the proposed CVPL values suggest the true number of cluster. Comparing $L_1$ with $L_2$ penalty in Scenario I, the outperformance of $k=2$ among all the other clusters are higher with $L_1$ penalty, which results from the sparsity of the two mean structures. When the sample size decreases as in Scenario II, such pattern becomes less obvious. It shows that the smaller dimension of images attenuates the discrepancy between $L1$ and $L_2$ regularizations. In the setting of Nuclear regularization, the proposed CVPL value leads to the true number of clusters, which is due to the low rank of mean structures. 
\begin{table}[h]
	\centering
	\caption{The cross validated penalized likelihood (CVPL) values obtained from different number of clusters and penalties under two scenarios. }
	\label{sim1:table1}
	\begin{tabular}{cccccccc}
		\hline \hline
		\multirow{2}{*}{Penalty} & \multicolumn{1}{l}{\multirow{2}{*}{$\lambda$}} & \multicolumn{3}{c}{CVPL (Scenario I)} &\multicolumn{3}{c}{CVPL (Scenario II)}  \T\B \\ \cline{3-8}
		& \multicolumn{1}{l}{}                                    & $k=2$            &         $k = 3$     &     $k = 4$      & $k=2$            &         $k = 3$     &     $k = 4$      \T   \\ \hline
		\multirow{3}{*}{L1}                    & 0.5                                                     &           $2.345^*$      &       2.337    &      2.333      &        ${0.458^*}$         &     0.453    &       0.451 \T  \B \\
		& 1                                                      &        ${ 2.344^*}$         &     2.336        &        2.330     &        ${0.457^* }$         &     0.455   &       0.452      \T \B \\
		& 1.5                                                       &         ${2.341^*} $         &          2.337  &    2.332     &         ${0.458^*} $         &          0.457   &   0.455     \T \B   \\ \hline
		\multirow{3}{*}{L2}      & 0.5                                                       &            ${2.351^*}$       &        2.349    &      2.344                  &            ${0.462^*}$       &        0.449    &     0.431    \T   \B \\
		& 1                                                       &        ${ 2.352^*}$           &         2.350   &       2.345   &        ${ 0.450^*}$           &         0.434   &      0.419   \T  \B \\
		& 1.5                                                       &        ${2.352^*} $          &           2.349  &        2.344          &        ${0.446^*} $          &           0.429  &      0.413  \T \B  \\ \hline
		\multirow{3}{*}{Nuclear} & 0.5                                                       &      ${2.351^*}$             &     2.348      &     2.343                           &      ${0.461^*}$             &   0.456       &     0.452    \T \B \\
		& 1                                                       &         ${2.351^*}$          &     2.348       &     2.344    &         ${0.461^*}$          &     0.457      &    0.452    \T \B  \\
		& 1.5                                                       &        ${2.353^*}$           &       2.349      &   2.345              &        ${0.460^*}$           &       0.456   &  0.454     \T \B    \\ \hline\hline
		\multicolumn{8}{l}{* The highest values across different scenarios ($\times 10^5$)}  \T  \B
	\end{tabular}
\end{table}

\subsection{Results on Comparing with K-Means}
This section is contributed to compare the proposed approach with K means. Similar to Section~\ref{section:sim1}, we generated signals using the same mean and covariance structures. In Scenario III, the sample size is set to be $50$ and the dimension of images $20 * 20$. In Scenario IV, we increase the sample size to $100$ and the dimension to $60 * 60$. To compare the results obtained from the two underlying approaches, we calculate the adjusted random index \citep{milligan1986study} and accuracy. We repeat the procedure $200$ times for this simulation study.

Results are summarized in Table~\ref{sim2:table2}. In Scenario III where the size is relatively low, the benefit of the proposed method is critical compared to K means. The ARI and accuracy values are almost double of the results obtained from K means. When it comes to larger sample size, which is presented as Scenario IV, the gain is also apparent. Among all the regularizations, the $L1$ penalty performs superiously due to the sparsity of the generated signals. 
\begin{table}[H]
	\caption{The adjusted random index (ARI) and accuracy obtained from the proposed method and K means under Scenario III and IV.}
	\label{sim2:table2}
	\begin{small}
		\hskip-1.0cm
		\begin{tabular}{cccccccccc}
			\hline \hline
			\multirow{2}{*}{Penalty} & \multicolumn{1}{l}{\multirow{2}{*}{$\lambda$}} & \multicolumn{2}{c}{ARI (Scenario III)}                                  & \multicolumn{2}{c}{Accuracy}   &\multicolumn{2}{c}{ARI (Scenario IV)}                                  & \multicolumn{2}{c}{Accuracy}                      \T \B          \\ \cline{3-10} 
			& \multicolumn{1}{l}{}                        & \multicolumn{1}{l}{our method} & \multicolumn{1}{l}{kmeans} & \multicolumn{1}{l}{our method} & \multicolumn{1}{l}{kmeans}                   & \multicolumn{1}{l}{our method} & \multicolumn{1}{l}{kmeans} & \multicolumn{1}{l}{our method} & \multicolumn{1}{l}{kmeans}                           \T \B \\ \hline
			\multirow{4}{*}{L1}      & 0                                           & 0.867                 & \multirow{4}{*}{0.513} & 0.882                     & \multirow{4}{*}{0.626}  
			& 0.644                    & \multirow{4}{*}{0.517} & 0.696                 & \multirow{4}{*}{0.607} 
			\T \B \\
			& 0.5                                         & 0.924                     &          &           0.938                      &                & 0.691                  &                            & 0.744                 &          
			\T \B       \\
			& 1                                           & 0.962                      &                            & 0.980                     &        &
			0.781                  &                            & 0.822                   &      
			\T \B            \\
			& 1.5                                         & 0.966                      &                            & 0.985                   &         & 0.788                   &                            & 0.824               &                     \T \B           \\ \hline
			
			\multirow{3}{*}{L2}      & 0.5                                         & 0.879                      &                            & 0.892                        &     & 0.632                &                            & 0.687                    &         \T \B                 \\
			& 1                                           & 0.907                    & 0.514                 & 0.918                 & 0.623     & 0.665                  & 0.518               & 0.715                     & 0.607      \T \B          \\
			& 1.5                                         & 0.868                      &                            & 0.881              &                   & 0.788                   &                            & 0.824             &           \T \B             \\ \hline
			\multirow{3}{*}{Nuclear} & 0.5                                         & 0.898                         &                            & 0.909                     &                 & 0.645                &                            & 0.697                   &              \T \B                 \\
			& 1                                           & 0.860                    & 0.515                & 0.876              & 0.623      & 0.660                    & 0.516           & 0.710                 & 0.607       \T \B            \\
			& 1.5                                         & 0.884                    &                            & 0.897                 &                 & 0.636                    &                            & 0.687             &                  \T \B           \\ \hline \hline
		\end{tabular}
	\end{small}
\end{table}

\section{Analysis of Odor Memory Data}
\label{section:real1}
In this section, we focus on analyzing a LFP dataset from a memory coding experiment on non-spatial events \citep{allen2016nonspatial}. Rats were trained to identify a series of five odors during the experiment. Each of the odors was presented through an odor port. In most of the cases, those five odors were in the same sequence (``{\it in-sequence}" odors) while there were some violations (``{\it out-sequence}" odors). For example, odor sequence {\it ABCDE} is an ``in-sequence" odor yet {\it ABBDE} is an ``out-sequence" odor. Rats were required to poke and hold their nose in the port to correctly identify whether the odors were ``in" or ``out" sequence. Throughout the experiment, spike and LFP data were collected. 22 electrodes were implanted in the CA1 pyramidal layer of the dorsal hippocampus, among which we only focus on 12 electrodes exhibiting task-critical single-cell activity. The whole LFP dataset contains 247 trials with a sampling rate 1000 Hertz and $T=2000$ time points. Figure~\ref{real:ts} exposes a snapshot of the LFP signals across 12 electrodes. 
\begin{figure}[th!] \centering
	\includegraphics[scale = 0.8]{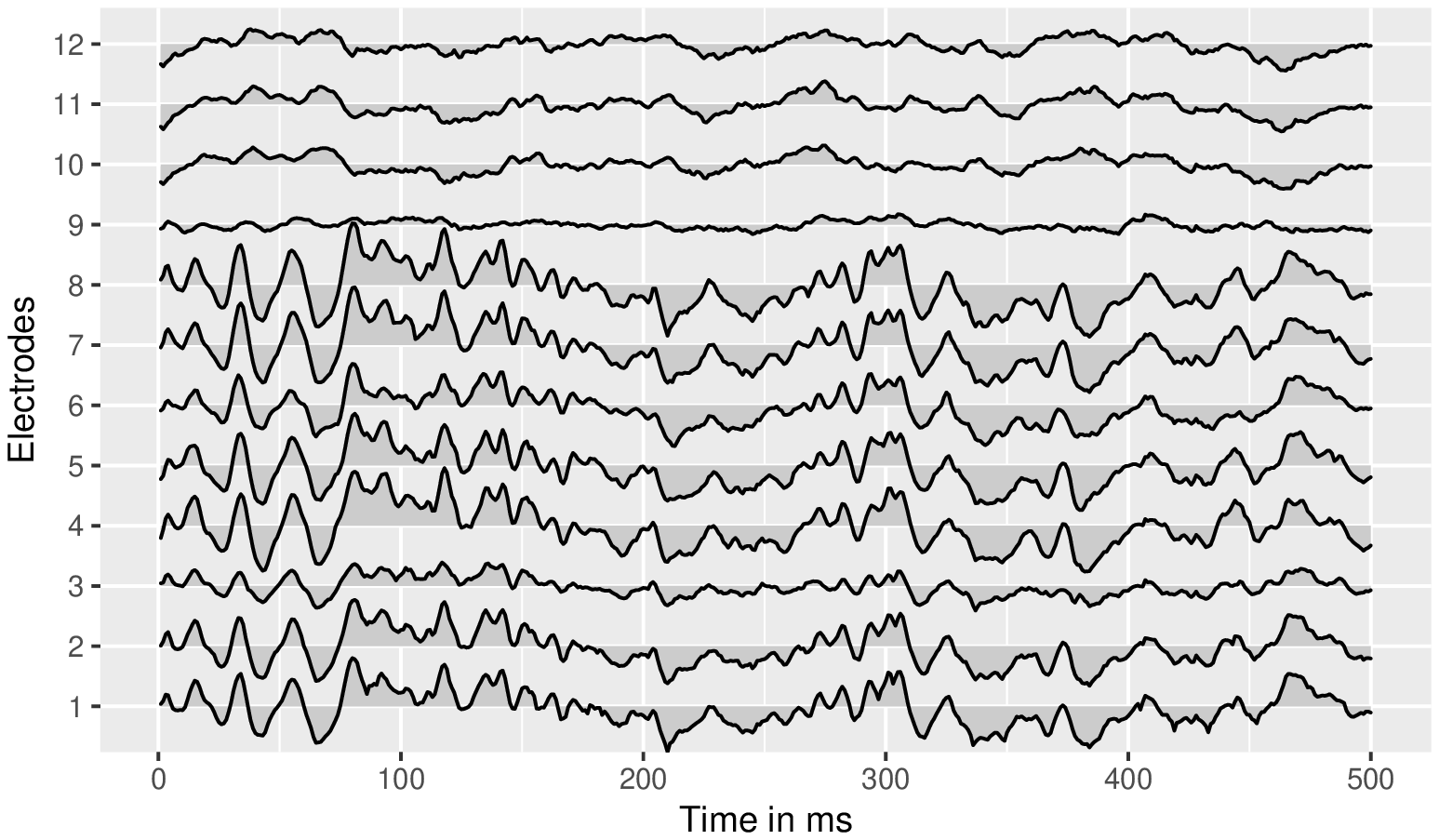} 
	\caption{Time series plot of LFP signals across 12 electrodes in trial 1. The plot only presents the first 500 time points.}
	\label{real:ts}
\end{figure}

\subsection{Time Domain Analysis on Imaging Clustering}
\label{section:real_time}
We applied the proposed clustering method to the LFP dataset with 247 trials to identify underlying patterns. As an initial step, we focus on time domain to uncover the association between raw multi-channel signals with ``in-sequence" or ``out-sequence" patterns. 
We implemented the proposed method to the raw LFP signals across all the 247 trials. 
\begin{table}[H]
	\centering
	\caption{The cross validated penalized likelihood values and the adjusted random index obtained across different number of clusters among all the three penalties.}
	\label{real:table1}
	\begin{tabular}{ccccccc}
		\hline \hline
		\multirow{2}{*}{Penalty} & \multicolumn{1}{l}{\multirow{2}{*}{$\lambda$}}   &\multicolumn{3}{c}{CVPL} &  \multicolumn{2}{c}{ARI} \T\B \\ \cline{3-7}
		& \multicolumn{1}{l}{}                                    & $k=2$            &         $k = 3$     &     $k = 4$ & our method & K means \T \B  \\ \hline
		\multirow{3}{*}{L1}      &0 &  1.290* & 1.285 & 1.281  & 0.768 &  \T \B \\              & 0.5                                                     &           1.253    &       1.253*    &    1.246  & 0.786 &   \multirow{3}{*}{0.499}  \T  \B \\
		& 1                                                      &        1.243*        &    1.206     &       1.204& 0.768  &  \T \B \\
		& 1.5                                                       &        1.249*      &      1.234 &    1.218 & 0.780 &  \T \B   \\ \hline
		\multirow{3}{*}{L2}      & 0.5                                                       &   1.302*         &        1.107  &    1.240  & 0.768 & \multirow{4}{*}{0.510} \T   \B \\
		& 1                                                       &      1.301*      &       1.027  &     1.202  & 0.774 & \T  \B \\
		& 1.5                                                       &      1.298*       &      1.189  &        1.235 &0.756 &  \T \B  \\ \hline
		\multirow{3}{*}{Nuclear} & 0.5                                                       &      1.309*     &     1.299     &  1.274  & 0.756 & \multirow{4}{*}{0.498} \T \B \\
		& 1                                                       &         1.299*        &    1.287   &   1.277 & 0.733 &  \T \B  \\
		& 1.5                                                       &      1.290*         &       1.286     &   1.214  &0.711 & \T \B    \\ \hline \hline
		\multicolumn{7}{l}{* The highest CVPL value ($\times 10^5$). } \T \B\\
	\end{tabular}
\end{table}

Table~\ref{real:table1} summarizes the cross validated penalized likelihood values among different number of clusters and penalties. It is obvious that 2 clusters are mostly suggested especially in the case of L2 or nuclear norm regularization. These findings motivate us to further investigate the cluster results with respect to the ``in/out sequence" patterns. Table~\ref{real:table1} shows such association. The adjusted random index was related to the true  label of ``in/out sequence" patterns. Comparing to K means, the proposed method outperforms in detecting the latent structure representing ``in" or ``out" sequences. 
Filter the LFPs by all the ``in-sequence'' signals.

As a further step, researchers are also interested in understanding how LFP signals are related to rat's correctness in this experiment. Due to the small size of ``out" sequence trials, we only focus on those ``in" sequence trials. In this way, we are able to investigate on the ``sensitivity" (true positive rate) of the experiment. 
\begin{table}[H]
	\centering
	\caption{The cross validated penalized likelihood values obtained across different number of clusters on all the ``in-sequence" trials.}
	\label{real:table3}
	\begin{tabular}{cccccccc}
		\hline \hline
		\multirow{2}{*}{Penalty} & \multicolumn{1}{l}{\multirow{2}{*}{$\lambda$}}   &\multicolumn{4}{c}{CVPL}  &  \multicolumn{2}{c}{ARI}\T\B \\ \cline{3-8}
		& \multicolumn{1}{l}{}                                    & $k=2$            &         $k = 3$     &     $k = 4$  & $k=5$   & our method & K means \T   \\ \hline
		\multirow{4}{*}{L1}         &0 & 1.135* & 1.135* & 1.126 & 1.131  & 0.762 & \multirow{4}{*}{0.506 } \T \B \\         
		& 0.5                                                     &           1.103*     &       1.076   &    1.084 & 1.094* & 0.783 &  \T  \B \\
		& 1                                                      &        1.099*        &    1.070   &       1.077
		&1.136*& 0.783 &   \T \B \\
		& 1.5                                                       &        1.107*      &     1.1078 &    1.118* 
		& 1.068  &0.609 & \T \B   \\ \hline
		\multirow{3}{*}{L2}      & 0.5                                                       &   1.142*         &        1.139  &  0.885 &  1.144* &0.769 &\multirow{3}{*}{0.499 } \T   \B \\
		& 1                                                       &      1.139*      &       1.016  &     1.101*
		&  0.986 & 0.743 &  \T  \B \\
		& 1.5                                                       &      1.150*       &    0.865  &       1.016
		& 1.061* &0.762 &  \T \B  \\ \hline
		\multirow{3}{*}{Nuclear} & 0.5                                                       &      1.159*     &     1.125    &  1.119 &1.126* &0.769 & \multirow{3}{*}{0.498} \T \B \\
		& 1                                                       &        1.153*        &   1.116  &   1.136* &  1.105  & 0.756 & \T \B  \\
		& 1.5                                                       &      1.141*        &     1.142*   &   1.036 & 1.123 & 0.783 & \T \B    \\ \hline \hline
		\multicolumn{6}{l}{* The top two CVPL values ($\times 10^5$). \T \B}
	\end{tabular}
\end{table}
Table~\ref{real:table3} shows the cross validated penalized likelihood obtained from the proposed approach. Among all the regularizations and $\lambda$ values, $k=2$ stands out among all the possible clusters. These results inspire us to further study the consistency between cluster results and the ``correctness" of this experiment. Table~\ref{real:table3} also presents the adjusted random index in relation to the ``correctness" labels. Compared to K means, our proposed approach is able to successfully identify the rat's ``correctness" on identifying ``in/out" sequences. It is worth mentioning that in addition to 2 clusters, Table~\ref{real:table3} also suggests 5 clusters. These results indicate our approach can possibly identify the five different odors. We will shed light on this direction in the next section.
%\subsubsection{Cross Correlation Clustering}
\subsection{Time Frequency Clustering Analysis}
We will continue to uncover the latent structure carried from the LFP dataset.  \cite{allen2016nonspatial} suggests two oscillatory bands (Theta: 4 - 12 Hertz and Slow Gamma: 20 - 40 Hertz) yield strong power and playing significant roles in detecting the ``in/out" sequences. Figure~\ref{fig:real_tf} shows the time frequency plot on Theta and Slow Gamma bands. Although these two bands enjoy the most power, low frequency theta band apparently obtains much more than slow gamma bands. It has been shown that slow gamma bands were strongly modulated by the ``in" and ``out" pattern \cite{allen2016nonspatial}. In this study, to take one step further, we applied the proposed method to the spectrum of Theta and Slow Gamma bands separately.  
\begin{figure}[H]
	\centering
	\includegraphics[width = .8\textwidth, height = 0.4\textheight]{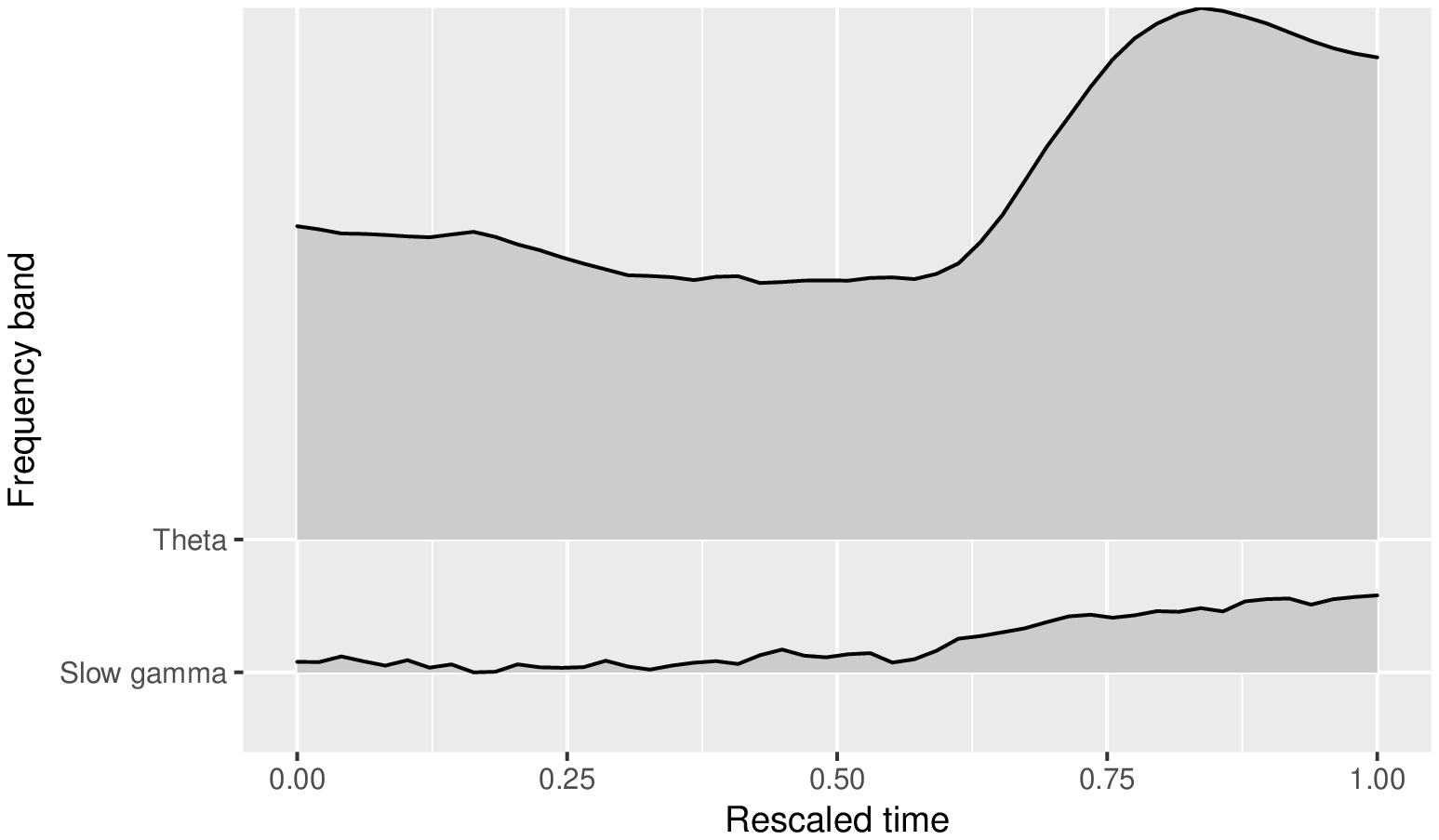}
	\caption{The time frequency plot of Theta and Slow Gamma bands over the ``in-sequence" trials.}
	\label{fig:real_tf}
\end{figure}
Table~\ref{real:table5} presents the results after implementing the proposed method to the spectrum on Theta band. It can be easily found that for each regularization setting, 4 or 5 clusters are highly suggested. We further compare the 5 cluster results with the true odor sequence. As is shown in Table~\ref{real:table5}, the consistency is strong especially when comparing wit K means. Our approach provides some evidence indicating the association between the low frequency band (Theta) and the odor sequence. 
\begin{table}[H]
	\centering
	\caption{The cross validated penalized likelihood obtained from the ``in-sequence" trials. The spectrum are from Theta band. }
	\label{real:table5}
	\begin{tabular}{cccccccc}
		\hline \hline
		\multirow{2}{*}{Penalty} & \multicolumn{1}{l}{\multirow{2}{*}{$\lambda$}}   &\multicolumn{4}{c}{CVPL}   &\multicolumn{2}{c}{ARI}     \T\B \\ \cline{3-8}
		& \multicolumn{1}{l}{}                                    & $k=2$            &         $k = 3$     &     $k = 4$  & $k=5$    & our method & K means \T   \\ \hline
		\multirow{3}{*}{L1}                    & 0                                                    &    11.001       &      11.300*    &  11.198*  & 11.172 &0.712 & \multirow{4}{*}{0.679}  \T  \B \\          & 0.5                                                     &           8.516     &      8.975*    &    8.849 & 8.997*  & 0.692 &  \T  \B \\
		& 1                                                      &        8.650     &   8.632     &      8.725*
		&8.745* & 0.703&   \T \B \\
		& 1.5                                                       &        8.571     &     8.705* &   8.556
		& 8.701* & 0.709 &  \T \B   \\ \hline
		\multirow{3}{*}{L2}      & 0.5                                                       &   8.965*         &        8.881* &  8.671 &  7.277 & 0.693 &\multirow{3}{*}{0.672} \T   \B \\
		& 1                                                       &      8.719*   &     8.388 &   8.544*
		&  7.616& 0.686 &  \T  \B \\
		& 1.5                                                       &      8.650     &    8.632  &       8.825*
		& 8.745* & 0.682 &  \T \B  \\ \hline
		\multirow{3}{*}{Nuclear} & 0.5                                                       &     9.034     &     9.196*    &  9.183 & 9.259* & 0.707 &\multirow{3}{*}{0.671}  \T \B \\
		& 1                                                       &       9.013        &   9.166  &   9.255* &  9.263* &0.714 & \T \B  \\
		& 1.5                                                       &     8.571      &     9.040*   &  8.995* & 8.969
		& 0.712 & \T \B    \\ \hline \hline
		\multicolumn{8}{l}{* The top two highest values ($\times 10^3$).} \T \B
	\end{tabular}
\end{table}

% comments start here
\iffalse
\begin{table}[H]
	\centering
	\caption{The adjusted random index in relation to the odor sequence. The spectrum are from Theta band.}
	\label{real:table6}
	\begin{tabular}{cccc}
		\hline \hline
		\multirow{2}{*}{Penalty} & \multicolumn{1}{l}{\multirow{2}{*}{$\lambda$}} & \multicolumn{2}{c}{ARI}                                    \T \B          \\ \cline{3-4} 
		& \multicolumn{1}{l}{}                        & \multicolumn{1}{l}{our method} & \multicolumn{1}{l}{kmeans}  \T \B \\ \hline
		\multirow{4}{*}{L1}      & 0                                           & 0.7115747              & \multirow{4}{*}{0.6788572}   \T \B \\
		& 0.5                                         & 0.6912991                  &                                    \T \B       \\
		& 1                                           & 0.7034477                   &                          \T \B            \\
		& 1.5                                         & 0.7087680               &                        \T \B           \\ \hline
		\multirow{3}{*}{Nuclear} & 0.5                                         & 0.7068409             &                                \T \B                 \\
		& 1                                           & 0.714423             & 0.6707302               \T \B            \\
		& 1.5                                         & 0.7119098                 &                                \T \B           \\ \hline \hline
	\end{tabular}
\end{table}
\fi
%% comments end here

Further, we concentrate on the Slow Gamma band. \cite{allen2016nonspatial} has established the conclusion that slow gamma band strongly aligned with the ``in/out" pattern. In this part, we applied the proposed method to all the ``in-sequence'' trials to uncover latent patterns. Table~\ref{real:table7} summarizes the cross validated penalized likelihood values among different clusters. 2 clusters are being recommended in most of the cases. We later compare the cluster result with the ``correctness" labels. In the case of nuclear norm regularization, the adjusted random index (0.5733) is almost $20\%$ higher than K means (0.497). 
\begin{table}[H]
	\centering
	\caption{The cross validated penalized likelihood obtained from the ``in-sequence" trials. The spectrum are from Slow Gamma band.}
	\label{real:table7}
	\begin{tabular}{cccccc}
		\hline \hline
		\multirow{2}{*}{Penalty} & \multicolumn{1}{l}{\multirow{2}{*}{$\lambda$}}   &\multicolumn{4}{c}{CVPL}  \T\B \\ \cline{3-6}
		& \multicolumn{1}{l}{}                                    & $k=2$            &         $k = 3$     &     $k = 4$  & $k=5$  \T   \\ \hline
		\multirow{3}{*}{L1}                    & 0.5                                                     &           8.470*   &      8.395    &    8.064 & 7.993 \T  \B \\
		& 1                                                      &        8.129*    &   8.023   &      7.507
		&7.312  \T \B \\
		& 1.5                                                       &       7.689*    &     7.641 &   7.215
		& 6.765 \T \B   \\ \hline
		\multirow{3}{*}{L2}      & 0.5                                                       &   8.360*         &        7.933 & 7.980  &  7.660 \T   \B \\
		& 1                                                       &      7.977*   &     7.755 &   5.744
		&  6.977 \T  \B \\
		& 1.5                                                       &     7.696     &    7.754*  &      6.584
		& 6.502 \T \B  \\ \hline
		\multirow{3}{*}{Nuclear} & 0.5                                                       &     8.687     &     8.785*    &  8.531 & 8.373 \T \B \\
		& 1                                                       &       8.686*     &   8.416  &   8.532 &  8.183 \T \B  \\
		& 1.5                                                       &     8.534*      &    8.438   & 8.324 & 7.981 \T \B    \\ \hline \hline
		\multicolumn{6}{l}{* The highest values ($\times 10^3$).} \T \B 
	\end{tabular}
\end{table}

%%%% comments start here
\iffalse
\begin{table}[H]
	\centering
	\caption{The adjusted random index in relation to the ``correctness" labels. The spectrum are from Slow Gamma band.}
	\label{real:table8}
	\begin{tabular}{cccc}
		\hline \hline
		\multirow{2}{*}{Penalty} & \multicolumn{1}{l}{\multirow{2}{*}{$\lambda$}} & \multicolumn{2}{c}{ARI}                                    \T \B          \\ \cline{3-4} 
		& \multicolumn{1}{l}{}                        & \multicolumn{1}{l}{our method} & \multicolumn{1}{l}{kmeans}  \T \B \\ \hline
		
		\multirow{3}{*}{Nuclear} & 0.5                                         & 0.51204390             &                                \T \B                 \\
		& 1                                           & 0.5293871            &0.4977168950              \T \B            \\
		& 2                                        & 0.57337355                 &                                \T \B           \\ \hline \hline
	\end{tabular}
\end{table}
\fi

\section{Analysis of Rat Stroke Data}
\label{section:real2}
In this section, we apply the proposed approach to another LFPs dataset from a rat stroke experiment. In this study, LFPs were recorded before and after the stroke. 32 electrodes were implanted with 4 layers shown in Figure~\ref{real2:schematic}. Throughout this section, we work on the signals of 5 minutes before and after the stroke. The sampling rate is 1000 Hertz and each epoch is 1 second long. One of the scientific interests from this experiment is to identify the ``latent" patterns that lead to before and after stroke. 
\begin{figure}[H] 
	\centering
	\includegraphics[scale = 0.35]{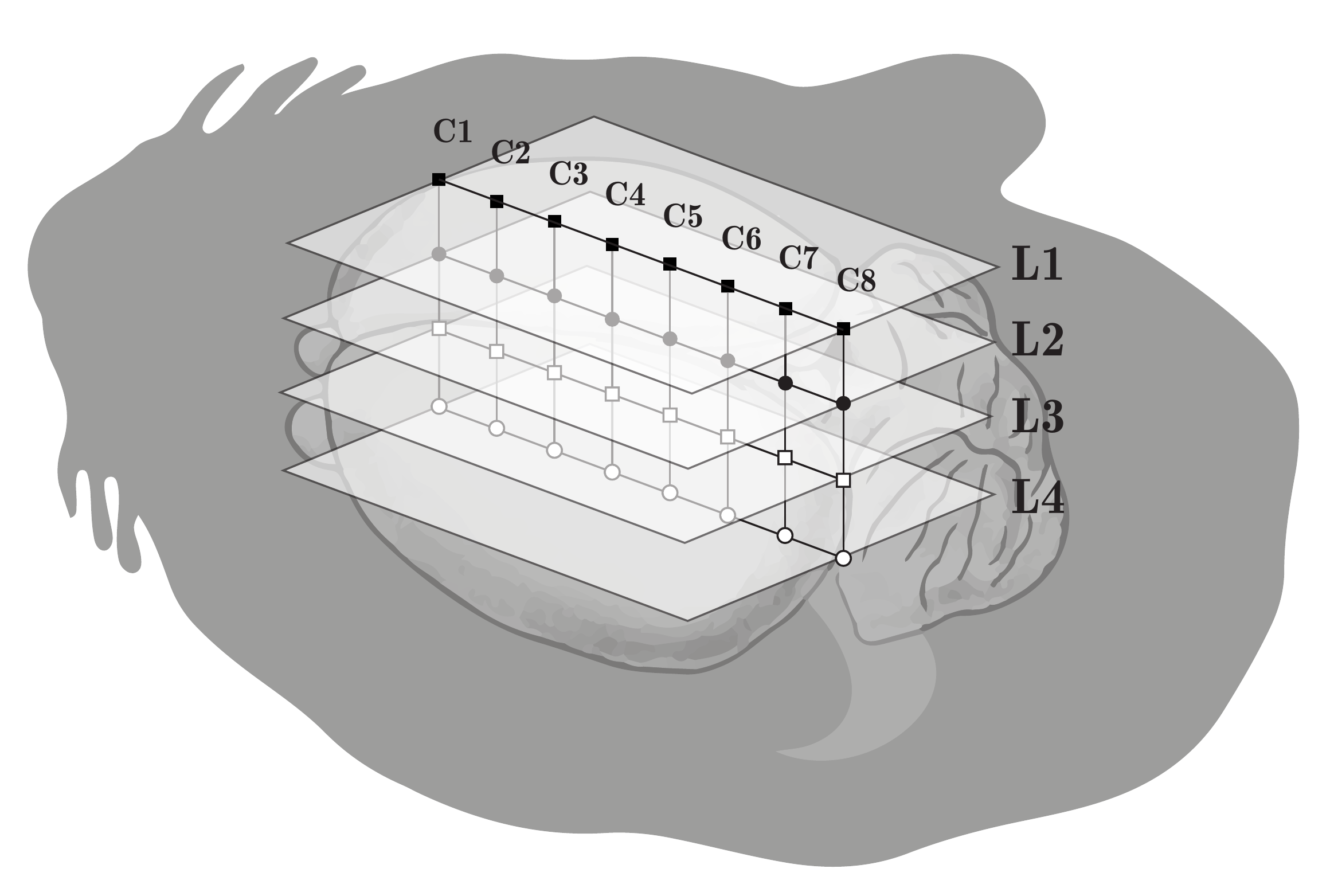} 
	\caption{The schematic diagram of electrodes implanted in rat brain.}
	\label{real2:schematic}
\end{figure}
As preliminary analysis, we implemented time frequency analysis on this dataset. Figure~\ref{real2:timefreq} shows the log power spectra of two typical channels. These results were obtained by averaging all the trials before and after stroke separately. Most of the channels behaves ``smoothly" within each epoch and there exists small discrepancy before and after stroke. However, just like the case of Channel 10, some channels presents nonnegligible dynamics and obvious difference between and after stroke. These findings shows that it is not optimal to average over or vectorize all the channels when we do cluster analysis to identify the ``latent" pattern before and after stroke.

\begin{figure}[H] 
	\begin{tabular}{cc}
			\hskip -2.1cm
		\includegraphics[scale = 0.35]{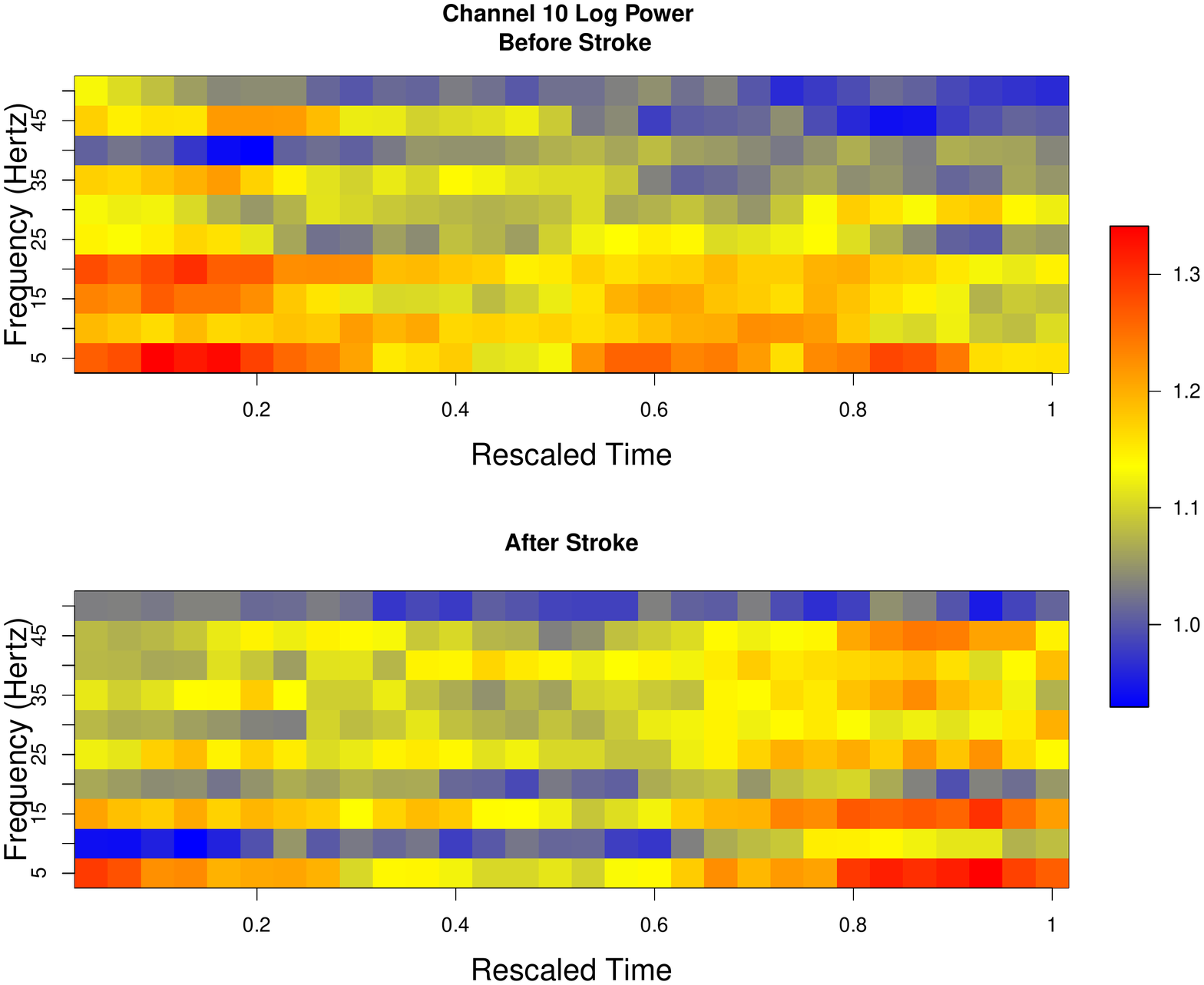} & 	\includegraphics[scale = 0.35]{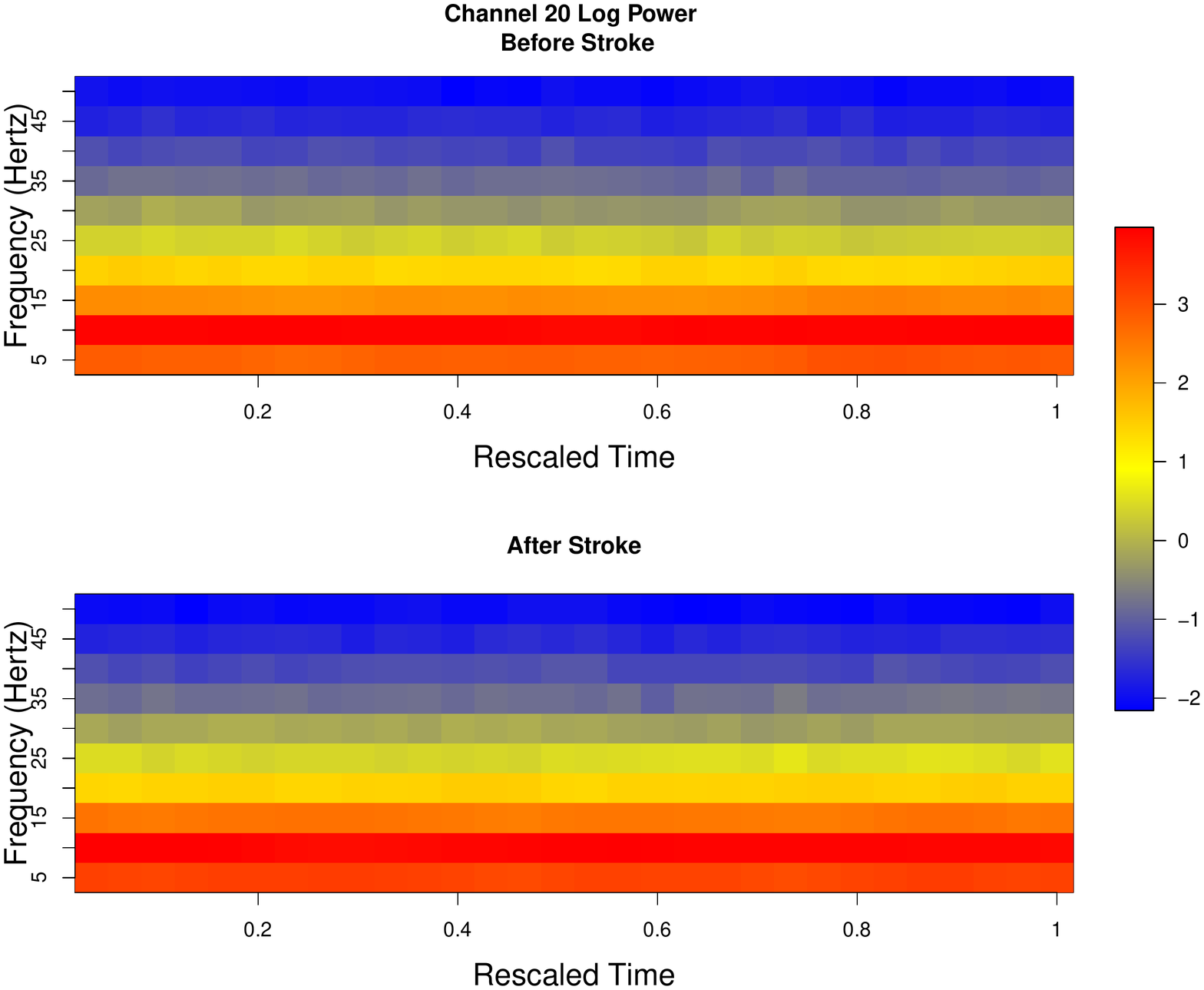}
	\end{tabular}
	\caption{The time frequency plot of Channel 10 and 20 among all the 600 trials before and after the stroke.}
	\label{real2:timefreq}
\end{figure}
To deepen the preliminary findings and motivate our proposed approach, we also study the dynamics across all the 32 channels before and after stroke. Figure~\ref{real2:timefreqtwobands} is the time frequency plot of Beta and Slow Gamma frequency bands across the channels. The log power spectra were obtained by averaging over the trials. Among the plots before and after stroke, we observe strong dependence across channels both for the two bands. This demonstrates the importance of introducing regularization terms into the mixture normal model. Comparing the plots before and after stroke, local discrepancy is easily identified. Such difference will be easily ignored if we just naively vectorize the original signals when doing cluster analysis. 
\begin{figure}[H] 
	\begin{tabular}{cc}
			\hskip -2.1cm
		\includegraphics[scale = 0.35]{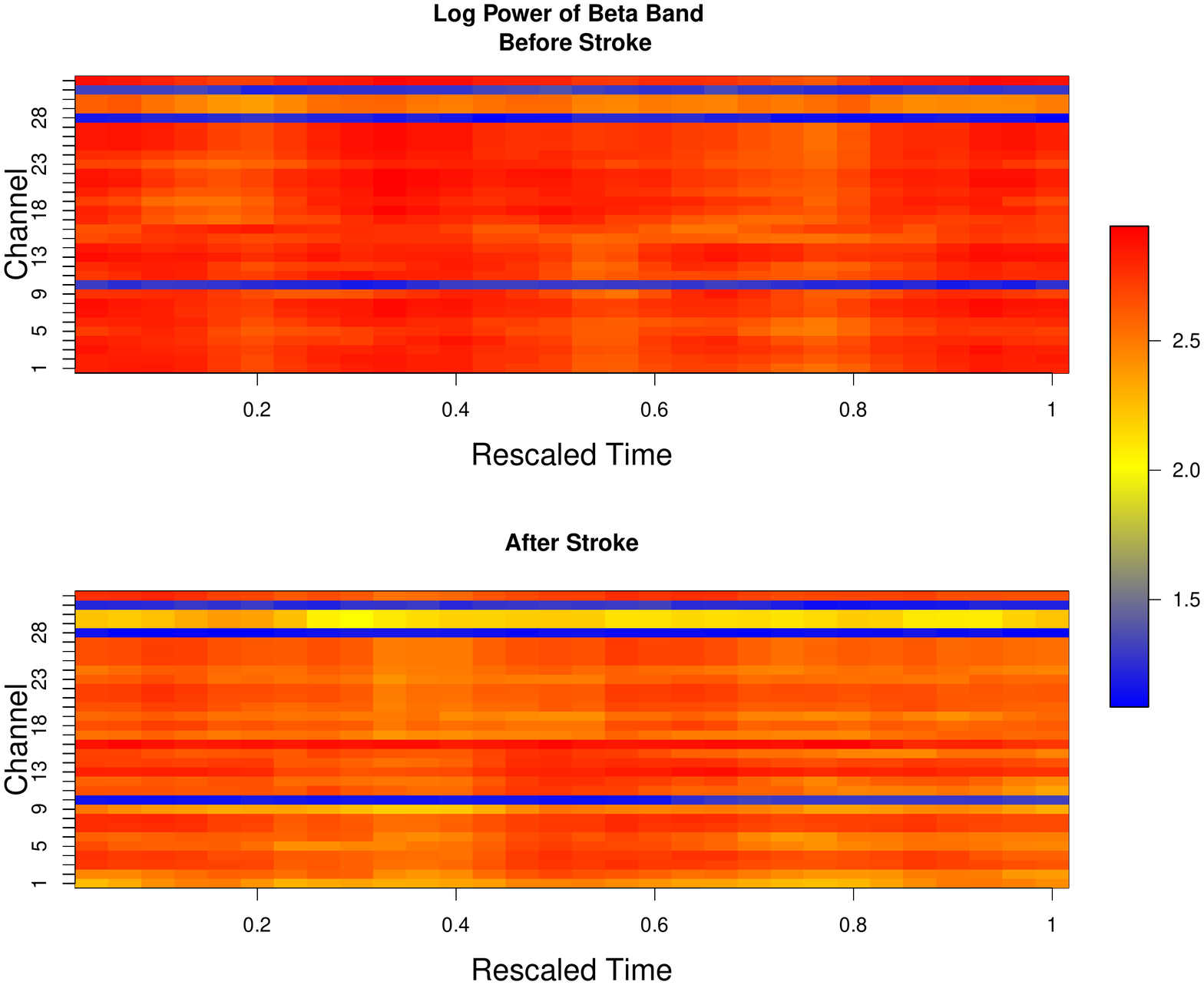} & 	\includegraphics[scale = 0.35]{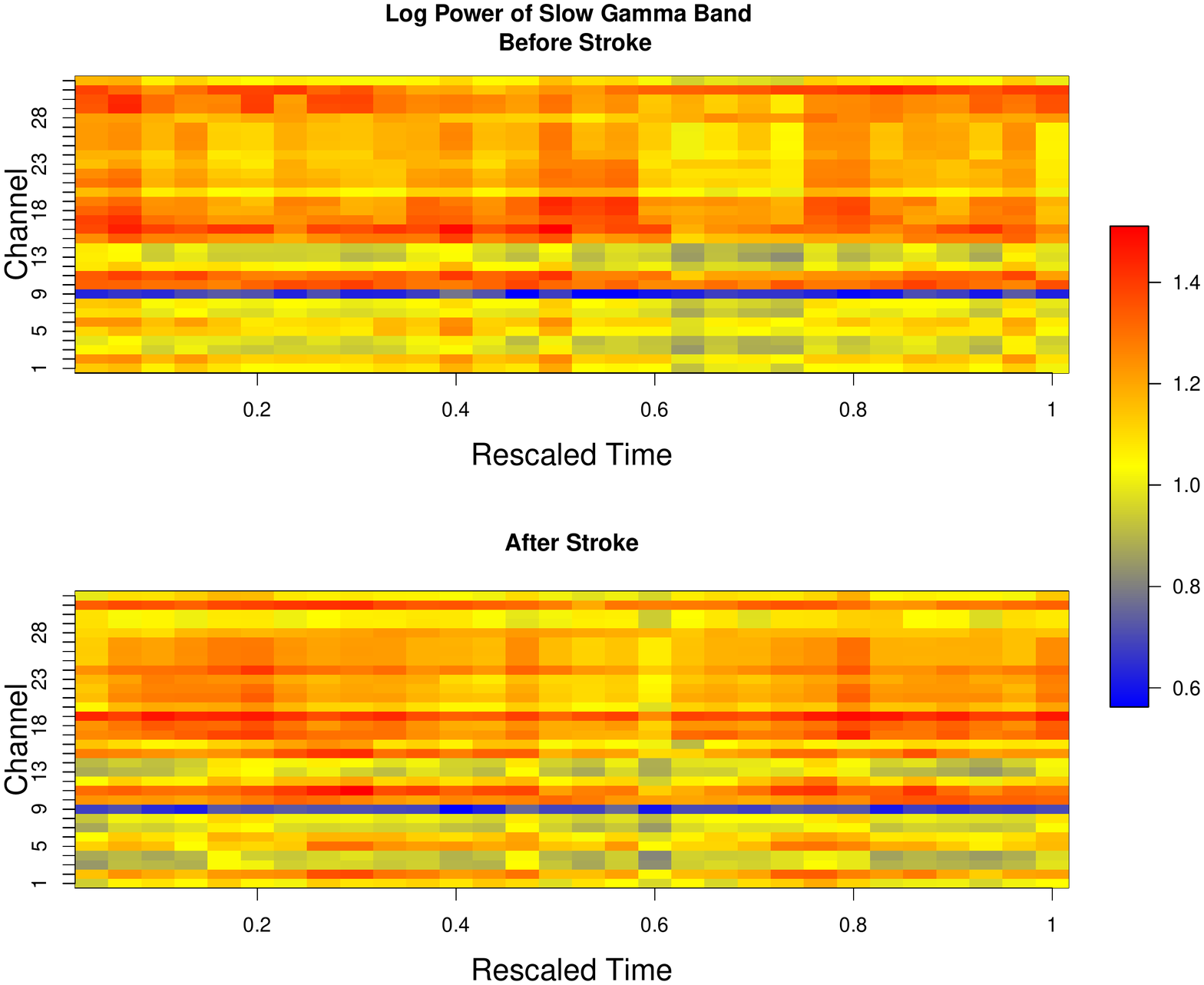}
	\end{tabular}
	\caption{The time frequency plot of particular frequency bands among all the channels before and after stroke.}
	\label{real2:timefreqtwobands}
\end{figure}
We applied the proposed approach to the time frequency images across all the trials before and after stroke. Table~\ref{real:table9} shows the cross validated penalized likelihood values across different number of clusters and regularizations. With only one exception, all the scenarios suggest 2 clusters. As the next step, we compare the 2 cluster results with the index related to ``stroke" or ``normal". Table~\ref{real:table10} summarizes the adjusted random index values (ARI). In comparing with K means results, the proposed approach outperforms in identifying ``stroke" or ``normal" sequences. Note that as by introducing regularizations, the proposed method is able to improve the results by 80\%. In particular, Slow Gamma bands performs perfectly (ARI 1.000) when adding nuclear norm term with $\lambda = 2$. This result is almost double the case without penalty (ARI 0.507). Similar pattern can also be found in Beta band case. These findings are consistent with the conjecture in preliminary analysis. 
\begin{table}[H]
	\centering
	\caption{The cross validated penalized likelihood obtained from all the trials. The log power spectra are from Beta and Slow Gamma bands.}
	\label{real:table9}
	\begin{tabular}{cccccccccc}
		\hline \hline
		\multirow{2}{*}{Penalty} & \multicolumn{1}{l}{\multirow{2}{*}{$\lambda$}}   &\multicolumn{4}{c}{CVPL (Slow Gamma)} &\multicolumn{4}{c}{CVPL (Beta)}    \T\B \\ \cline{3-10}
		& \multicolumn{1}{l}{}                                    & $k=2$            &         $k = 3$     &     $k = 4$  & $k=5$  & $k=2$            &         $k = 3$     &     $k = 4$  & $k=5$  \T   \\ \hline 
		\multirow{3}{*}{L1}                    & 0                                                   &           2.941  &      2.964*    &    2.764 & 2.822  & 4.645* & 4.627 & 4.526 & 4.598 \T  \B \\
		& 1                                                      &        2.472*    &   2.031   &     1.513
		&0.4213 & 3.98* & 3.268 & 3.594 & 1.676 \T \B \\
		& 2                                                       &       2.106*    &     1.370 &  1.288
		& 0.621 & 4.167* & 3.373 &3.227 & 3.277 \T \B   \\ \hline
		\multirow{3}{*}{L2}      & 0.5                                                       &   2.688*         &        2.474 & 2.306   &  2.184& 4.245* & 4.179 & 4.036 &3.424 \T   \B \\
		& 1                                                     &   2.484*         &        2.188 & 1.787  &  1.895 &4.063* &3.557 & 3.429 & 3.329 \T  \B \\
		& 2                                                       &     2.338*     &    2.163  &      1.539
		& 1.733& 4.024*&3.699& 2.972& 3.206 \T \B  \\ \hline
		\multirow{3}{*}{Nuclear} & 0.5                                                       &     2.806*    &     2.627    &  2.502 & 2.303& 4.464* &4.299 &4.130 & 3.963 \T \B \\
		& 1                                                       &       2.556*    &   2.362  &   1.946 &  1.720 & 4.191* &3.977 & 3.618  & 3.371\T \B  \\
		& 2                                                       &     2.748*      &    1.689  & 1.257 & 0.684  &3.687* & 3.274& 2.795 & 2.262 \T \B    \\ \hline \hline
		\multicolumn{10}{l}{* The highest values over different frequency bands ($\times 10^4$)} \T \B
	\end{tabular}
\end{table}

\begin{table}[H]
	\centering
	\caption{The adjusted random index in relation to ``Stroke". The spectrum are from Slow Gamma and Beta bands.}
	\label{real:table10}
	\begin{tabular}{cccccc}
		\hline \hline
		\multirow{2}{*}{Penalty} & \multicolumn{1}{l}{\multirow{2}{*}{$\lambda$}} & \multicolumn{2}{c}{ARI (Slow Gamma)}    &\multicolumn{2}{c}{ARI (Beta)}                                   \T \B          \\ \cline{3-6} 
		& \multicolumn{1}{l}{}                        & \multicolumn{1}{l}{our method} & \multicolumn{1}{l}{kmeans}   & \multicolumn{1}{l}{our method} & \multicolumn{1}{l}{kmeans}        \T \B \\ \hline
		\multirow{4}{*}{L1}      & 0                                           & 0.507             & \multirow{4}{*}{0.751} & 0.887 & \multirow{4}{*}{0.716} \T \B \\
		& 0.5                                         & 0.981                &                &  0.942   &                \T \B       \\
		& 1                                           & 0.961                   &                      & 0.914 &    \T \B            \\
		& 2                                         & 0.951           &                   &0.861&     \T \B           \\ \hline
		\multirow{3}{*}{L2} & 0.5                                         & 0.951           &                & 0.941 &                \T \B                 \\
		& 1                                           & 0.951             & 0.751    & 0.878&     0.716     \T \B            \\
		& 2                                         & 0.961            &              &0.787 &                  \T \B \\    \hline
		\multirow{3}{*}{Nuclear} & 0.5                                         & 0.951            &                &0.941 &                \T \B                 \\
		& 1                                           & 0.960           & 0.751       & 0.942&  0.715     \T \B            \\
		& 2                                         & 1.000                &             &0.951 &                   \T \B           \\ \hline \hline
	\end{tabular}
\end{table}

\section{Concluding remarks}
In this paper, we proposed a regularized probabilistic clustering framework to analyze matrix data. Compared to the existing approaches such as K means, the advantages are as follows: (1.) through working directly on matrices, we are able to capture the row-wise and column-wise correlation simultaneously; (2.) by introducing penalty terms into the likelihood function, the framework is able to ``uncover" the certain sparsity nature originated from the images or signals; (3.) The proposed approach provides theoretical justification as well as straightforward interpretability with low computational cost. 

Although this paper provides some promising results, analyzing matrix data is still a open ended problem. For instance, in the current work, choosing the number of clusters highly rely on some pre-specified measures (CVPL). As an extension, we could introduce a Bayesian framework into the clustering analysis to obtain a more data driven and interpretable optimal number of clusters.

%\section*{Acknowledgement}

%  If your paper refers to supplementary web material, then you MUST
%  include this section!!  See Instructions for Authors at the journal
%  website http://www.biometrics.tibs.org

%
%  \bibliographystyle{biom} 
% \bibliography{mybibilo.bib}
%\bibliographystyle{biom}
%\begin{thebibliography}{}

\bibliographystyle{chicago}
\bibliography{xg}

%\end{thebibliography}
\end{document}